\documentclass[a4paper,11pt]{article}
\usepackage[utf8]{inputenc}
\usepackage[T1]{fontenc}
\usepackage{geometry}
\usepackage{titling}
\usepackage{color}
\usepackage{babel}
\usepackage{refstyle}
\usepackage{float}
\usepackage{amsmath}
\usepackage{amsthm}
\usepackage{makecell}
\usepackage{bbm}
\usepackage{amssymb}
\usepackage{stmaryrd}
\usepackage{setspace}
\usepackage{comment}
\usepackage{mathtools}
\usepackage{subcaption}
\usepackage{caption}
\usepackage{tikz}
\usetikzlibrary{decorations.markings}
\usepackage{rotating} \usepackage{fontawesome} \usepackage[misc]{ifsym} \usepackage{ifthen}
\usepackage{graphicx}
\usepackage{tikz-3dplot}
\usepackage{xstring}
\usepackage{xcolor}
\usepackage[colorinlistoftodos]{todonotes}
\usepackage[unicode=true,pdfusetitle,bookmarks=true,bookmarksnumbered=false,bookmarksopen=false,breaklinks=false,pdfborder={0 0 1},backref=false,colorlinks=true,linkcolor=blue,urlcolor=blue]{hyperref}
\usepackage[nameinlink]{cleveref}

\usetikzlibrary{calc}

\usepackage{mleftright}
\mleftright \usepackage{soul}
\clearpage{}\pgfmathsetmacro{\scaling}{0.4}

\tikzset{plane/.style n args={3}{insert path={#1 -- ++ #2 -- ++ #3 -- ++ ($-1*#2$) -- cycle}},
unit xy plane/.style={plane={#1}{(\scaling,0,0)}{(0,\scaling,0)}},
unit xz plane/.style={plane={#1}{(\scaling,0,0)}{(0,0,\scaling)}},
unit yz plane/.style={plane={#1}{(0,\scaling,0)}{(0,0,\scaling)}},
get projections/.style={insert path={let \p1=(\scaling,0,0),\p2=(0,\scaling,0)  in 
[/utils/exec={\pgfmathtruncatemacro{\xproj}{sign(\x1)}\xdef\xproj{\xproj}
\pgfmathtruncatemacro{\yproj}{sign(\x2)}\xdef\yproj{\yproj}
\pgfmathtruncatemacro{\zproj}{sign(cos(\tdplotmaintheta))}\xdef\zproj{\zproj}}]}},
pics/unit cube/.style={code={
\path[get projections];
\draw (0,0,0) -- (\scaling,\scaling,\scaling);
\ifnum\zproj=-1
 \path[3d cube/every face,3d cube/xy face,unit xy plane={(0,0,0)}]; 
\fi
\ifnum\yproj=1
 \path[3d cube/every face,3d cube/yz face,unit yz plane={(\scaling,0,0)}]; 
\else
 \path[3d cube/every face,3d cube/yz face,unit yz plane={(0,0,0)}]; 
\fi
\ifnum\xproj=1
 \path[3d cube/every face,3d cube/xz face,unit xz plane={(0,0,0)}]; 
\else
 \path[3d cube/every face,3d cube/xz face,unit xz plane={(0,\scaling,0)}]; 
\fi
\ifnum\zproj>-1
 \path[3d cube/every face,3d cube/xy face,unit xy plane={(0,0,\scaling)}]; 
\fi
}},
3d cube/.cd,
xy face/.style={fill=lightgray},
xz face/.style={fill=darkgray},
yz face/.style={fill=gray},
every face/.style={draw,very thick}
}

\usetikzlibrary{patterns}
\tikzset{
        hatch distance/.store in=\hatchdistance,
        hatch distance=10pt,
        hatch thickness/.store in=\hatchthickness,
        hatch thickness=1pt
}
\makeatletter
\pgfdeclarepatternformonly[\hatchthickness]{my horizontal lines}
{\pgfpointorigin}{\pgfqpoint{100pt}{1pt}}{\pgfqpoint{100pt}{3pt}}
{
    \pgfsetcolor{\tikz@pattern@color}
    \pgfsetlinewidth{\hatchthickness}
    \pgfpathmoveto{\pgfqpoint{0pt}{0.5pt}}
    \pgfpathlineto{\pgfqpoint{100pt}{0.5pt}}
    \pgfusepath{stroke}
}

\pgfdeclarepatternformonly[\hatchthickness]{my vertical lines}
{\pgfpointorigin}{\pgfqpoint{1pt}{100pt}}{\pgfqpoint{3pt}{100pt}}
{
    \pgfsetcolor{\tikz@pattern@color}
    \pgfsetlinewidth{\hatchthickness}
    \pgfpathmoveto{\pgfqpoint{0.5pt}{0pt}}
    \pgfpathlineto{\pgfqpoint{0.5pt}{100pt}}
    \pgfusepath{stroke}
}

\pgfdeclarepatternformonly[\hatchthickness]{my north east lines}
{\pgfqpoint{-1pt}{-1pt}}{\pgfqpoint{4pt}{4pt}}{\pgfqpoint{3pt}{3pt}}
{
    \pgfsetcolor{\tikz@pattern@color}
    \pgfsetlinewidth{1pt}
    \pgfpathmoveto{\pgfqpoint{0pt}{0pt}}
    \pgfpathlineto{\pgfqpoint{3.1pt}{3.1pt}}
    \pgfusepath{stroke}
}

\pgfdeclarepatternformonly[\hatchthickness]{my north west lines}
{\pgfqpoint{-1pt}{-1pt}}{\pgfqpoint{4pt}{4pt}}{\pgfqpoint{3pt}{3pt}}
{
    \pgfsetcolor{\tikz@pattern@color}
    \pgfsetlinewidth{1.2pt}
    \pgfpathmoveto{\pgfqpoint{0pt}{3pt}}
    \pgfpathlineto{\pgfqpoint{3.1pt}{-0.1pt}}
    \pgfusepath{stroke}
}

\colorlet{grayX}{black!26!white}
\colorlet{darkgrayX}{black!47!white}
\makeatother

\makeatletter

\AtBeginDocument{}
\AtBeginDocument{}
\AtBeginDocument{}
\AtBeginDocument{}

\theoremstyle{plain}
\newtheorem{thm}{\protect\theoremname}
\theoremstyle{plain}
\newtheorem*{thm*}{\protect\theoremname}
\theoremstyle{definition}
\newtheorem{defn}[thm]{\protect\definitionname}
\theoremstyle{plain}
\newtheorem{cor}[thm]{\protect\corollaryname}
\theoremstyle{remark}
\newtheorem*{note*}{\protect\notename}
\theoremstyle{remark}
\newtheorem*{rem*}{\protect\remarkname}
\theoremstyle{remark}
\newtheorem*{notation*}{\protect\notationname}
\theoremstyle{remark}
\newtheorem{claim}[thm]{\protect\claimname}
\theoremstyle{plain}
\newtheorem{lem}[thm]{\protect\lemmaname}
\theoremstyle{plain}
\newtheorem{problem}{\protect\probname}
\theoremstyle{plain}
\newtheorem{prop}[thm]{\protect\propositionname}
\theoremstyle{plain}

\crefname{equation}{}{}

\makeatother

\providecommand{\claimname}{Claim}
\providecommand{\corollaryname}{Corollary}
\providecommand{\definitionname}{Definition}
\providecommand{\lemmaname}{Lemma}
\providecommand{\probname}{Problem}
\providecommand{\notationname}{Notation}
\providecommand{\notename}{Note}
\providecommand{\propositionname}{Proposition}
\providecommand{\remarkname}{Remark}
\providecommand{\theoremname}{Theorem}
\providecommand{\observationname}{Observation}

\makeatletter

\newif\if@mainmatter \@mainmattertrue
\newcommand\frontmatter{\cleardoublepage
  \@mainmatterfalse
  \pagenumbering{Roman}}
\newcommand\mainmatter{\cleardoublepage
  \@mainmattertrue
  \pagenumbering{arabic}}
\makeatother

\def\wasyfamily{\fontencoding{U}\fontfamily{wasy}\selectfont}
\DeclareTextFontCommand{\textwasy}{\wasyfamily}
\DeclareSymbolFont{wasy}{U}{wasy}{m}{n}
\SetSymbolFont{wasy}{bold}{U}{wasy}{b}{n}
\def\hexag{\mbox{\wasyfamily\char55}}

\newcommand{\FF}{\mathsf{F}}
\newcommand{\ff}{\mathsf{f}}
\newcommand{\N}{\mathbb{N}}
\newcommand{\Z}{\mathbb{Z}}
\newcommand{\R}{\mathbb{R}}
\newcommand{\I}{\mathcal{I}}

\newcommand{\vp}{\varphi}
\DeclareMathOperator{\lip}{Lip}
\DeclareMathOperator{\hexagon}{\hexag}

\newcommand{\norm}[1]{\left\lVert#1\right\rVert}

\newcommand{\hyp}[1]{$\mathcal{H}(#1)$}
\newcommand{\p}{^\prime}
\newcommand{\D}{\Delta}
\clearpage{}

\date{\today}
\author{Arye Deutch\thanks{Einstein Institute of Mathematics, Hebrew University of Jerusalem, Israel. \url{ arye.deutsch@mail.huji.ac.il}}
\and
Ohad Noy Feldheim\thanks{Einstein Institute of Mathematics, Hebrew University of Jerusalem, Israel. \url{ ohad.felheim@mail.huji.ac.il}} 
\and
Rani Hod\thanks{Blavatnik School of Computer Science, Raymond and Beverly Sackler Faculty of Exact Sciences, and Iby and Aladar Fleischman Faculty of Engineering, Tel Aviv University, Tel Aviv, Israel. Research supported by Len Blavatnik and the Blavatnik Family foundation. \url{ranihod@tau.ac.il}}}

\begin{document}

\title{Multi-layered planar firefighting}

\maketitle

\begin{abstract}
Consider a model of fire spreading through a graph; initially some vertices are burning, and at every given time-step fire spreads from burning vertices to their neighbors. 
The firefighter problem is a solitaire game in which a player is allowed, at every time-step, to protect some non-burning vertices (by effectively deleting them) in order to contain the fire growth. How many vertices per turn, on average, must be protected in order to stop the fire from spreading infinitely?

Here we consider the problem on $\Z^2\times [h]$ for both nearest neighbor adjacency and strong adjacency.
We determine the critical protection rates for these graphs to be~$1.5h$ and~$3h$, respectively. 
This establishes the fact that using an optimal two-dimensional strategy for all layers in parallel is asymptotically optimal.
\end{abstract} 
\section{Introduction}\label{sec:introintroduction}
lLet $G=\left(V,E\right)$ be an infinite graph.
The firefighter problem on $G$ is the following solitaire combinatorial game, introduced by Hartnell~\cite{hartnell1995firefighter}. The game starts with a finite starting set of \emph{burning} vertices $B(0)\subset V$. 
At every turn $t\in\N$, the player chooses an arbitrary finite collection of non-burning vertices and \emph{protects} them permanently, subject to the constraint that at most~$\ff(t)$ vertices are protected by time~$t$.
Then, the unprotected neighbors of burning vertices become burning. The goal of the game is to ensure that eventually no new burning vertices are generated.
If this goal is achieved, we say that the player \emph{contained} the fire.
We say that $C(G)$ is the \emph{critical protection rate} if for $\ff(t)=C(G) t$, there exists a starting set of burning vertices for which there is no play strategy that allows the player to contain the fire, while for $\ff(t)=(C(G)+\varepsilon) t$ for any $\varepsilon>0$, such a strategy exists for all finite $B_0\subset V$.

Denote the set of all positive integers by $\N$ and the set of~$h$ smallest positive integers by $[h] := \{1, 2, \ldots, h\}$, for $h\in\N$.
Write $\Box$ (resp., $\boxtimes$) for the Cartesian product (resp., the strong Cartesian product) of graphs. We study the firefighter problem on $G_1=G_1^h:=(\Z\;\Box\;\Z)\;\Box\;[h]$
and $G_2 = G_2^h:=(\Z\boxtimes\Z)\;\Box\;[h]$. These two infinite graphs have the same set $\Z\times\Z\times[h]$ of vertices, which we think of as $h$ vertically-stacked copies of the horizontal plane $\Z^2$. The difference is that the degree of a typical vertex in $G_1$ is $6$ (four horizontal neighbors, one above and one below), versus degree $10$ in $G_2$ ($8$ horizontal neighbors, one above and one below).

Our main result is the following.
\begin{thm}\label{thm:1}
For every $h\in\mathbb{N}$ and for $q\in\{1,2\}$, we have $C(G_q)=\frac32qh$.
\end{thm}
The special case $h=1$ of Theorem~\ref{thm:1} has been obtained by the second and third authors in~\cite{feldheim20133}.
In light of this, Theorem~\ref{thm:1} can be interpreted as a ``parallel repetition'' statement: in the multi-layered setting, the player cannot asymptotically improve upon the simple strategy of repeating the same two-dimensional strategy~$h$ times.
In fact, in the course of the proof we will see that the bound~$3qh$ is also viable for the graph $(\Z\boxtimes\Z)\boxtimes[h]$.
We believe that this observation holds in much greater generality (see discussion in~\Cref{subsec:OpenProblems}).  
\subsection{Background}
The firefighter problem is a model for the spread of a phenomenon at the face of an effort to contain it. The model is best suited to describe the spread of an epidemic (or a false rumor) in a population while a vaccine (or clear contradictory evidence) are administered to prevent it.

The problem was introduced by Hartnell~\cite{hartnell1995firefighter} in 1995, formulated with one firefighter per time-step and a single vertex as initial fire and was generalized to the version described above. The problem could be seen as an earlier variant of Conway's celebrated \emph{angel problem} \cite{conway1996angel}, in which the angel has power $k=1$, the devil is oblivious to the location of the angel and must be certain to catch it in order to win the game. 

On finite graphs, the most natural challenges are to reduce the number of burning vertices as the process terminates and to reduce the speed at which the fire is contained, see~\cite[Chapter~5]{Fin09}.
The problem has been also studied from an algorithmic point of view. 
MacGillivray and Wang~\cite{MacGillivray2003Ontheff} show that on bipartite graphs finding a strategy which minimizes the number of burning vertices is NP-complete. This has been extended to trees and cubic graphs in~\cite{finbow2007firefighter,king2010firefighter}.
On trees, however, the solution could be approximated up to a constant factor in polynomial time (See~\cite{adjiashvili2018firefighting,Approx2,hartke2004attempting,HL00,iwaikawa2011improved}).
For a survey on the problem see~\cite{Fin09} and references therein.

On infinite, vertex-transitive graphs with planar-type growth, the most natural question is to recover the critical protection rate. 
Currently, the only known values of~$C(G)$ are $C(\Z\;\Box\;\Z)=\frac{3}{2}$ and $C(\Z\boxtimes\Z)=3$, obtained in~\cite{feldheim20133} (following earlier results by Wang and Moller~\cite{moeller2002fire}, Ng and Raff~\cite{ng2005fractional} and Messinger~\cite{messinger2008average}).
These rely on an isoperimetric argument developed by Fogarty in her thesis~\cite{fogarty2003catching}.
On the 6-regular hexagonal lattice we have
$1\le C(\Z\;\hexagon\;\Z)\le2$ \cite{fogarty2003catching}, while on the cubic triangular lattice $\frac12\le C(\Z\;\triangle\;\Z)\le1$ (unpublished, could be obtained in a similar way).
Computing the critical protection rate of these lattices remains an open problem.

Develin and Hartke~\cite{DH07} generalized Fogarty's argument to show that~$2d-1$ firefighters are required to contain a single-source fire in~$\Z^{\Box d}$ and conjectured that for $\ff(t) \ll t^{d-2}$, there exists an initial fire such that any strategy using $\ff(t)$ firefighters fails to contain it.
A new result by Amir, Baldasso and Kozma~\cite{amir2020firefighter} use a Fogarty-type argument together with an isoperimetric tool to prove a generalized version of Develin and Hartke's conjecture, namely that in Cayley graphs with polynomial growth, any strategy with $\ff(t) \ll t^{d-2}$ cannot contain a large enough initial fire.

Only when~$G$ has a planar growth rate, constant protection rate should suffice to contain the fire and $C(G)$ is well defined. Thus, in a constant protection rate setting, the multi-layered variant studied here is a natural three-dimensional analog of the problem.
We further believe that our methods could be applied to nearly any local connectivity structure between multiple layers of $\Z\;\Box\;\Z$ and $\Z\boxtimes\Z$, so that whenever a three-dimensional layered graph has its layers connected along a sub-lattice, playing identically on each layer is asymptotically optimal.
We also see the new tools presented here as a stepping stone for tackling other graphs, and hopefully also a continuous variant (see \Cref{subsec:OpenProblems} below).  
\subsection{Outline of the proof for the lower bound}\label{subsec:pf outline}
\textbf{Firefronts and levels.} 
The analysis in~\cite{feldheim20133} of the fire evolution in the planar lattice considered four \emph{firefronts}, forming a rectangle that bounded the fire at every time step. To generalize this to the multi-layered setting, we separately treat the horizontal components of each firefront, which we call \emph{levels}. Each level is a horizontal segment of vertices, and the four levels in the each layer form a rectangle. 
The firefronts are defined in a way which guarantees that each front will always be adjacent to many burning vertices.

Much like in~\cite{feldheim20133}, the orientation of the levels differs by $45^\circ$ between $G_1$ and $G_2$: levels are parallel to the horizontal axes in $G_2$, while in $G_1$ they are diagonal to them. 
In each front, the horizontal distances between the origin and each level may differ; however, a front cannot have ``vertical jumps''; i.e., the horizontal distance to the origin differs by at most one between vertically adjacent levels of the same front.
We call this structure the \emph{fronts structure} (see \Cref{fig:front_structure}).

\textbf{Movement of the levels.} The horizontal distance between the origin and a level can never decrease, only increase by one or stay the same. 
When we decide to increase this distance in a certain direction, we say that that level \emph{advances}, and otherwise that is it \emph{still}.
We say that a level is \emph{burning} (at some time) if at least one vertex on it is burning.
A burning level that advances is called \emph{active}. The only possible way for a non-burning level to advance is when ``pulled'' by an adjacent advancing level (burning or not) in order to maintain the ``no vertical jumps'' property. 
To show that the fire cannot be contained, we will prove that the fronts structure advances indefinitely.
In previous works, burning levels would always advance (see~\cite{feldheim20133}). In this work, we may decide to keep a burning level still for some time. This new approach allows us to better control the interplay between the the number of burning vertices on a level and the length of horizontally adjacent levels.

\textbf{Fierity and Potential.} We keep count of protected and burning vertices encountered by the fronts structure. 
This could be thought of as a time change of the original process which is non-homogeneous in space, in the sense that different fronts undergo a different time change.
We call the number of burning vertices on the $k$-th level of the $i$-th front its \emph{fierity} and denote it by~$\vp^k_i$.
When a level is burning, surely in the next time step the horizontal neighbors of its burning vertices are either burning or protected. An argument of Fogarty~\cite{fogarty2003catching} uses this observation to show that sum of the fierity $\vp^k_i$ of a burning level and the number~$f^k_i$ of newly counted protected vertices increases by at least~$q$ (when playing on $G_q$). 
In our situation vertices may shift from one firefront to an adjacent front; to keep track of these transitions we count the difference between vertices shifted into the $k$-th level of the $i$-th front and out of it, which we denote by $p^k_i$.

Putting these together we define~$\mu^k_i$, the \emph{potential} of the $k$-th level of the $i$-th front, as the sum of a level's fierity, protected and shifted vertices. A variant of Fogarty's argument is used to show that for an active level this number increases by at least~$q$.

Much of the effort would have been spared if we could have just followed the footsteps of~\cite{feldheim20133} and inductively prove that at any given time, there are at least~$3h$ active levels.
Via the corresponding potential increase, it would show that the fire cannot be contained.
Unfortunately, this is not always the case; indeed, there could be times at which there are fewer than~$3h$ active levels. 
We show, that when such a time occurs, this gap is compensated within $2h$ time-steps. Therefore, the average increase in the number of burning vertices on the front structure is still~$3h$.
    
\textbf{Overflow.} The source of additional growth in the fierity of a level comes from vertical \emph{overflow}. 
That is a situation in which the fire spreads from burning vertices above or below some level to that level. 
Overflow occurs when an active level is pulling another level. 
In order to guarantee sufficient overflow, we must allow a level to be active only when it has sufficiently many vertices, so that if it overflows to a neighbor, it will \emph{significantly} increase its fierity. 

\noindent\textbf{Remark.}
We strive to provide a unified proof for $q=1,2$, instead of duplicating the proof with minor changes.
In~\cite{feldheim20133} a proof for $q=2$ sufficed to imply the case $q=1$ using the observation that $\Z\boxtimes\Z$ is contained in the square of $\Z\;\Box\;\Z$, so playing on $\Z\boxtimes\Z$ is at least as easy for the player as playing on $\Z\;\Box\;\Z$ with twice as many firefighters.
For $h>1$, however, the connection between $G_1$ and $G_2$ is not as obvious, since the square of $G_1$ also allows double vertical steps.

Upon a first reading, the reader is advised to focus on~$G_2$ and ignore the case $q=1$.  
\subsection{Outline of the paper}\label{subsec:paper outline}
\Cref{sec:Preliminaries} consists of the definitions required to formulate the firefighter problem and the evolution of the firefronts structure given an activity criterion, concluding with a reduction of \Cref{thm:1} to a lower bound on the potential of the firefronts structure (\Cref{prop:inductive-step}).
In \Cref{sec:Fire growth} we define a finer notion of potential, which treats each firefront separately. We further establish key inequalities controlling the growth of this potential (\Cref{prop:mu>a>0} and~\Cref{prop:mu>vp}).
\Cref{sec:Main proof} is dedicated to the proof of \Cref{prop:inductive-step} and auxiliary lemmata. 
\subsection{Discussion}\label{subsec:Discussion}
In this paper we equip $\Z^2\times [h]$ with vertical connectivity of a Cartesian product, which is weaker than the strong Cartesian product. Hence we have $C(\Z\;\Box\;\Z\boxtimes[h])=\frac{3}{2}h$ and $C(\Z\boxtimes\Z\boxtimes[h])=3h$ as an immediate corollary of~\cref{thm:1}.
Indeed, the statement of~\cref{thm:1} only implies a lower bound on the critical protection rate, but the same containment strategy used to prove~\cref{thm:1}, which is repeating a horizontal containment strategy in each layer, would work here too (and also for any intermediate vertical connectivity).
We believe that one can use the same methods to obtain the same critical protection rates for any weaker regular connectivity, that is, a connectivity in which in every $k\times k$ horizontal square there is a vertex connected to the vertex above it (except in the topmost layer). For example, one can consider a connectivity in which $(x_1,y_1,z_1),(x_2,y_2,z_2)$ with $z_1\ne z_2$ are neighbors, if and only if $|z_1-z_2|=1,(x_1,y_1)=(x_2,y_2)$ and $x_1+y_1$ is divisible by~$3$. 
We also tend to believe that even if this connectivity is further weakened by subdividing every vertical edge into a path constant length, the same result holds. 
To deal with weaker connectives one needs to change the definition of a front structure and the activity of the fronts in order to guaranteed overflow. We have decided not to do to, in order to increase the readability of our work.

The main technical innovation of the paper, is that the analysis of the fire is not done by estimating the number of burning vertices on the rectangular boundary of the burning region at every time-step, but rather on the rectangular boundary of an artificially chosen smaller domain. This allows the classical techniques from potential theory, to ignore thin stretches of burning vertices which reduce the number of burning site on the rectangular boundary of the burning region. 

For $h=1$, the lower bound of~\cite{feldheim20133} is quite tight: for $\ff(t)\le \frac32q t$ it is impossible to contain even the smallest possible initial fire --- one burning vertex.
Clearly, this no longer makes sense for $h>1$. Indeed, the maximal degree in~$G_q$ is~$2+4q$, so it is possible to contain, in a single time-step, any initial fire of size
\[
    \frac14h \le \frac{\frac32qh}{2+4q}.
\]
Furthermore, it is possible to have an initial fire of size at least $\frac1{32}q h^3$
that can be contained in~$h$ time-steps for $\ff(t) = \frac32qht$. Assume for simplicity that $h$ is divisible by~$8/q$, and set $r=\frac34h-1$.
The pyramid-shaped construction (see~\Cref{fig:pyramid}) goes as follows: in the bottom layer we have a square (resp., diamond) of side (resp., diagonal) length~$r$, which forms the base of the pyramid. The choice of~$r$ makes it possible to horizontally surround the base in a single time-step. Above the base we have a smaller square/diamond, of side/diagonal length~$r-2$, which grows to a square/diamond of side/diagonal length~$r$ in the second time-step. Again, there are just enough firefighters to horizontally surround the fire on layer~$2$, by placing them directly above the layer~$1$ firefighters. By time~$h$, the cylinder of firefighters will reach the top layer, and the fire will be contained.

It is quite easy to see that for $q=2$ the initial fire consists of
\[
        r^2 + (r-2)^2 + (r-4)^2 + \cdots
    =   \frac{r(r+1)(r+2)}6
    =   \frac{(r+1)^3}6 - \frac{r+1}6
    =   \frac9{128}h^3 - \frac18h
    \ge \frac1{16}h^3
\]
burning vertices, and for $q=1$ there are
\[
        \frac{1+r^2}2+\frac{1+(r-2)^2}2+\cdots+\frac12
    =   \frac{r(r+1)(r+2)}{12} + \frac{r+1}4
    >   \frac{(r+1)^3}{12} 
    =   \frac9{256} h^3
    >   \frac1{32} h^3
\]
burning vertices.
\begin{figure}[h!]
    \centering
    \captionsetup{margin=1cm}
	\tdplotsetmaincoords{65}{35}
	\begin{tikzpicture}[line join=round]
    \begin{scope}[tdplot_main_coords]
    \path[get projections];
	\foreach \z/\a/\b/\c/\d in { 0/ 2/ 2/ -2/ -2
                                         ,1/ 1/ 1/ -1/ -1}
            {
			\pgfmathsetmacro{\aa}{\a-1}
                
                \foreach \xi in {\d,...,\b}
                {
                        \path (\scaling*\xi,\scaling*\a,\scaling*\z) pic{unit cube};
                }
                \foreach \yi in {\aa,...,\c}
                {
                        \path (\scaling*\d,\scaling*\yi,\scaling*\z) pic{unit cube};
                }
                \foreach \xi in {\d,...,\b}
                {
                        \path (\scaling*\xi,\scaling*\c,\scaling*\z) pic{unit cube};
                }
                \foreach \yi in {\a,...,\c}
                {
                        \path (\scaling*\b,\scaling*\yi,\scaling*\z) pic{unit cube};
                }
            }
        \path (\scaling* 0,\scaling* 0,\scaling* 2) pic{unit cube};
	\end{scope}
    \end{tikzpicture}
    \tdplotsetmaincoords{65}{55}
    \begin{tikzpicture}[line join=round]
    \begin{scope}[tdplot_main_coords]
        \path[get projections];
        \path (\scaling*-2,0,0) pic{unit cube};
        \foreach \y in {1,0,-1}
            {\path (\scaling*-1,\scaling*\y,0) pic{unit cube};}
        \foreach \y in {2,1,0,-1,-2}
                    {\path (0,\scaling*\y,0) pic{unit cube};}
        \foreach \y in {1,0,-1}
            {\path (\scaling,\scaling*\y,0) pic{unit cube};}
        \path (\scaling*2,0,0) pic{unit cube};
        
        \path (\scaling*-1,0,\scaling) pic{unit cube};
        \foreach \y in {1,0,-1}
            {\path (0,\scaling*\y,\scaling) pic{unit cube};}
        \path (\scaling,0,\scaling) pic{unit cube};
        \path (0,0,\scaling*2) pic{unit cube};
        
    \end{scope}
    \end{tikzpicture}
    \caption{
        An initial fire in~$G_q$ of size at least $qh^3/32$ that can be contained with~$\frac32qh$ firefighters per turn.
        Here~$h=8$, $r=5$, and indeed $5^2+3^2+1^2 = 35 > 8^3 / 16$ for $q=2$ and $13+5+1 = 19 > 8^3 / 32$ for $q=1$.
    }
    \label{fig:pyramid}
\end{figure}
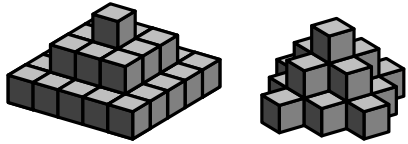
\subsection{Open problems}\label{subsec:OpenProblems}
In the section we present what we see as the most interesting problems concerning the firefighter problem on lattices with planar type growth, which we hope that this new technique could help tackling.

\begin{problem}\label{prob:pr1}
    Given any local connectivity for $\Z\times\Z$, compute the critical protection rate.
    In particular, compute \emph{$C(\Z\;\hexagon\;\Z)$} and  $C(\Z\;\triangle\;\Z)$.
\end{problem}

For $\Z\;\hexagon\;\Z$, at first sight the problem seems similar to $\Z\boxtimes\Z$, and indeed several authors claimed to have solved it only to find an error in their arguments. 
Even the second and third authors have claimed in~\cite{feldheim20133} that this problem seems approachable using the same techniques. 
However, the fact that the total length of three adjacent hexagonal fronts is not necessarily half of the circumference makes the current analysis fail. 
It appears that part of the reason that current methods fall short of tackling the hexagonal lattice, is the fact that they player can create a situation in which there are many burning vertices on the convex boundary of the burning domain, but not on the bounding hexagon containing it. 
It seems worthwhile to try and extend the techniques presented here to take into account these burning vertices.

\textbf{Continuous variant.}
In 2007 Bressan~\cite{bressan2007differential} introduced the following continuous space-time variant of the problem, which we present here as a scaling limit of a continuous space-discrete time process.
Consider a starting fire~$F(0)$ that is a path-connected set in~$\mathbb{R}^2$ containing the origin. At every~$\varepsilon$ time-step, the player adds to the set of protected paths~$P(t)$, initialized as the empty set, a segment of length~$c \cdot \varepsilon$ (in some metric~$d_\text{protect}$), disjoint from the fire.
After this, the fire~$F(t+\varepsilon)$ is increased to be the connected component of origin in~$F(t)+B_\varepsilon\setminus P(t)$, where~$B_\varepsilon$ is the~$\varepsilon$ ball in some metric~$d_\text{fire}$.
Taking~$\varepsilon$ to~$0$, a continuous variant of the firefighter problem arises where~$C(d_\text{protect}, d_\text{fire})$ is the infinimum of the set of~$c$ for which the fire could be contained.
For example, $C(L_\infty, L_1) = 1.5$ and $C(L_1, L_\infty) = 3$ by taking scaling limit over the results of~\cite{feldheim20133}.

In the natural setting, when~$B_\varepsilon$ is the Euclidean ball and segment length is measured in Euclidean metric, the answer is conjectured by Bressan to be~$2$ and a tight upper bound is provided. We conjecture that the same holds for any~$L_p$ metric.
\begin{problem}\label{prob:pr2}
    In the continuous firefighter problem on $L_p$, compute $C(L_p, L_p)$.
\end{problem}

We hope that progress towards \Cref{prob:pr1} would render it possible, through analogy and through scaling limits, to obtain new results for \Cref{prob:pr2}. 
\section{Preliminaries}\label{sec:Preliminaries}
\textbf{Notation and conventions.}
The sub-lattice $\Z^2\times\{k\}$ is called the $k$-th \emph{layer} of our graph (either $G_1$ or $G_2$). We think of it as being horizontal.
Throughout the paper, we follow the convention that superscripts refer to layers (and levels), subscripts refer to directions (and fronts), and the argument of functions is time.
Moreover, we reserve the indices $i$ and $j$ for directions/fronts, the indices $k$ and $\ell$ for levels/layers, and the variables $s, t$ and $\tau$ for time.
Addition and subtraction in subscript indices is henceforth always taken modulo 4.

The notation $A\sqcup B$ refers to the union $A\cup B$ of two sets, $A$ and $B$, assumed to be \emph{disjoint}.

For a function $g$ on the integers (resp., on sets), denote by $\D g(t)$ its discrete derivative $\D g(t):=g(t)-g(t-1)$ (resp., $\D g(t):=g(t) \setminus g(t-1)$).
Sometimes it would be more convenient to directly define $\D g(t)$, and obtain $g(t)$ inductively as $g(t):=g(t-1) + \D g(t)$ for some base case $g(0)$ (resp., as $g(t):=g(t-1) \sqcup \D g(t)$).

Omitting an index (superscript or subscript) in a notation serves to represent summation (or union) over this index; e.g., $f_{i}(t):=\sum_{k\in\left[h\right]}f_i^k(t)$.
Additionally, writing~$\max$ (resp.,~$\min$) for an index represents taking the maximum (resp., minimum) over it; e.g., $r_i^{\min}(t):=\min_{k\in[h]}r_i^k(t)$.

Given a set of vertices $U\subset \Z^2\times[h]$, denote by~$U^+$ its closed neighborhood, namely vertices in $U$ or adjacent to some $u\in U$.

\textbf{Evolution of the process, given a player strategy.}
Given an initial set~$B(0) \subset\mathbb{Z}^{2}\times\left[h\right]$ of burning vertices, and an increasing function $\FF(t)_{t\in\N}$ of sets of vertices that the player would like to protect by time~$t$, we now specify the evolution of the process, by defining the set~$B(t)$ of vertices burning at time~$t$.

To avoid the question of precise timing (i.e., what happens first at time~$t$), we highlight the timeline of the game in the following table:

\begin{center}
\begin{tabular}{|c|c|c|}
\hline 
\textbf{Time} & \textbf{Event} & \textbf{Updated quantity}\\\Xhline{2\arrayrulewidth}
$-\frac{1}{3}$ & The initial fire is created & $B(0)$ is set\\[4pt] \hline 
$t-\frac{2}{3}$ & Additional vertices are protected & 
\begin{tabular}{@{}c@{}}
    The protected set becomes \\
    $\FF(t)\setminus B(t-1)$\\[2pt]
\end{tabular}
\\[4pt] \hline
$t-\frac{1}{3}$ & Fire spreads to unprotected adjacent vertices& $B(t)$ is determined\\[4pt] \hline
$t$ & Nothing &\\\hline 
\end{tabular}
\end{center}

We inductively define $B_{\FF}(t)=B\left(t\right)$ for all $t\in\N$ by
\begin{equation}\label{eq:main-evo}
    B(t)=B(t-1)\cup \left(B(t-1)^+\setminus \FF(t)\right).
\end{equation}
Observe that here we allow~$\FF(t)$ to overlap with~$B(t-1)$; however, such overlap has no effect on the evolution of~$B(t)$. 
The reader may think of~$\D \FF(t)$ as orders given by the player concerning which vertices to protect at time~$t$, so that ``illegal'' orders (i.e., orders to protect vertices that are already burning) are simply ignored. Hence, $\FF(t)\setminus B(t-1)$ are the vertices protected just before the fire spreads further at time $t-\tfrac{1}{3}$.

In~\Cref{subsec:fronts-evolution} we will define $F(t)$, which is a refinement of~$\FF(t)$,  consisting only of effectively protected vertices; that is, vertices that would have been burned up to time~$t$, but were protected.

\textbf{Directions.} 
First we define, in a clockwise fashion, the horizontal directions for the firefronts in~$G_1$ and~$G_2$ (see~\Cref{sec:introintroduction}).
For~$G_1$, these are 
\begin{align*}
\theta_{0}&:=(+\tfrac12,+\tfrac12,0), &
\theta_{1}&:=(+\tfrac12,-\tfrac12,0), \\
\theta_{2}&:=(-\tfrac12,-\tfrac12,0), &
\theta_{3}&:=(-\tfrac12,+\tfrac12,0); \\
\shortintertext{while for~$G_2$, these are}
\theta_{0}&:=(0,+1,0),&
\theta_{1}&:=(+1,0,0),\\
\theta_{2}&:=(0,-1,0),&
\theta_{3}&:=(-1,0,0).
\end{align*}
We denote by $\I:=\{0,1,2,3\}$ the index set of horizontal directions. Recall that we use these always modulo~$4$; in particular we write $|i-j|=1$ if $i$ and $j$ are consecutive modulo~$4$ (e.g., $i=3,j=0$). 

In addition, let~$\phi:=(0,0,1)$ denote the vertical unit vector.\footnote{Note the visual cue provided by the letter~$\theta$ (resp.,~$\phi$) to the vector's horizontal (resp., vertical) orientation.}

The (infinite) line of distance~$d$ from the origin, perpendicular to the $i$-th direction, in the~$k$-th layer,
can be written in terms of~$\{\theta_j: j\in\I\}$ and~$\phi$ as 
\[
  d\theta_i + \R\theta_{i+1} + k\phi = \{ d\theta_i + m\theta_{i+1} + k\phi: m\in\R\}.
\]
We will be concerned with vertices on finite segments of these lines, as the following notation captures.
For $a,b,d\in\N$, $i\in\I$ and $k\in[h]$, let
\[
    L^k_{i,d}(a,b):=\left\{d\theta_i+m\theta_{i+1}+k\phi\ :\ m \in [-a,b)\right\}\cap \Z^3.
\]
\begin{figure}[h!]
    \captionsetup{margin=1cm}
    \subfloat[$\scriptstyle{q=2},{d=2},{a=2},{b=3}$\label{fig:L1}]{
    \begin{tikzpicture}[scale=0.8]
        \draw[step=1cm,darkgray,thin] (-0.3,-3.3) grid (4.8,3.3);
        \draw[draw=none,darkgrayX, thick,fill=darkgrayX] (2,-2) rectangle (3,3);
        \draw [|-|, line width=0.3mm] (0.5,-3.1) -- +(2,0) node[midway,above] {$d=2$};
        \draw [|-|, line width=0.3mm] (3.3,0.525) -- +(0,2) node[midway,right] {$a=2$};
        \draw [|-|, line width=0.3mm] (3.3,0.475) -- +(0,-3) node[midway,right] {$b=3$};
        \draw [->] (0.5,0.5) -- +(0,0.3);
        \draw [->] (0.5,0.5) -- +(0.3,0);
        \foreach \x in {0, ..., 4} { \foreach \y in {-3, ..., 2} {
                \draw [gray] (\x + 0.5,\y + 0.5) circle (0.2mm);
            }
        }
    \end{tikzpicture}}
    \hfill
    \subfloat[$\scriptstyle{q=1},{d=3},{a=2},{b=5}$\label{fig:a=L2}]{
    \includegraphics{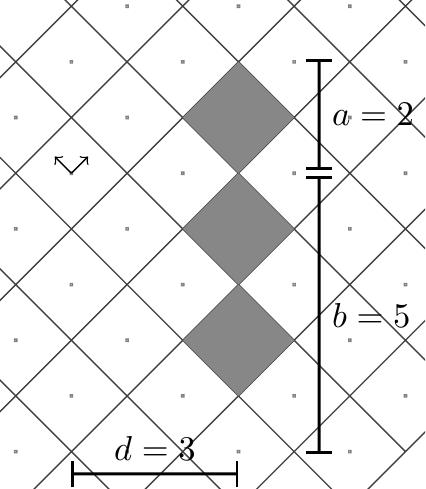}}
    \hfill
    \subfloat[$\scriptstyle{q=1},{d=4},{a=2},{b=5}$\label{fig:a=L3}]{
    \includegraphics{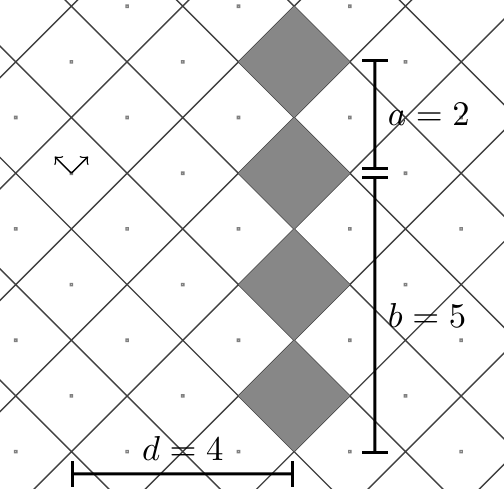}}
    \caption{Illustration of $L^k_{1,d}(a,b)$ (in dark gray) for various values of $q,d,a$ and $b$. The origin and the Cartesian axes are marked by short arrows. In each subfigure, $\theta_0$ points upwards, $\theta_1$ --- to the right, $\theta_2$ --- downwards and $\theta_3$ --- to the left.
    }
    \label{fig:fig_L}
\end{figure}
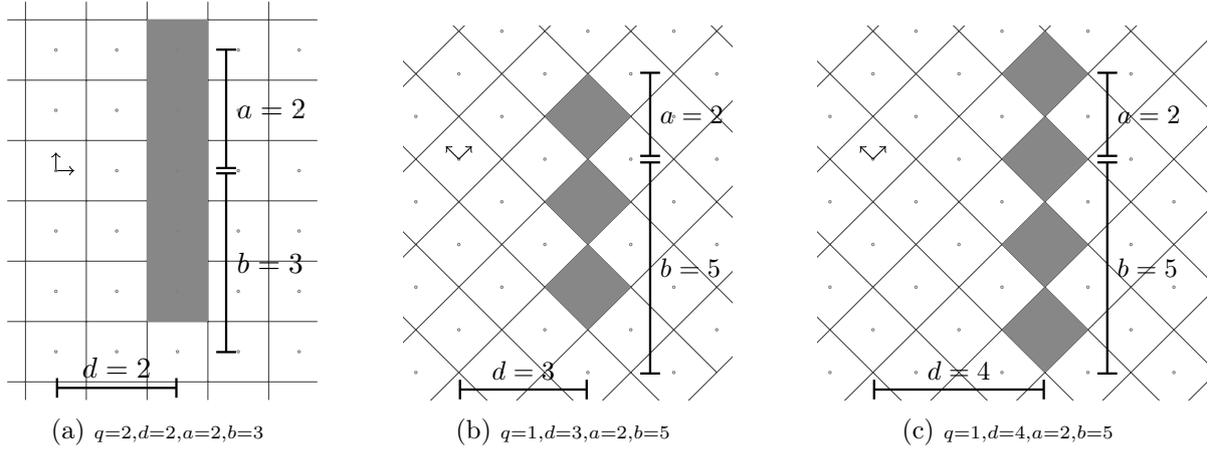 
For $G_2$ we simply have $|L^k_{i,d}(a,b)| = a+b$, but for $G_1$ the cardinality depends on the parities of $a+d$ and $b+d$ (see~\Cref{fig:fig_L}) and we have
\[
        |L^k_{i,d}(a,b)|
    =   \frac12(a+b) + \frac14(-1)^{a+d} - \frac14(-1)^{b+d}
    \le \frac12\left(a+b+1\right).
\]
For simplicity we use the unified bound
\begin{equation}\label{eq:Lab-length}
        |L^k_{i,d}(a,b)|
    \le \frac q2\left(a+b+2-q\right)
    \le \frac q2\left(a+b+1\right).
\end{equation}
\subsection{The fronts structure}
Given two vectors $\vec x=(x^1,\ldots,x^h),\vec y=(y^1,\ldots,y^h)$, we say that $\vec x$ \emph{dominates} $\vec y$ if $y^k\le x^k$ for all $k\in[h]$.
Call a vector $\vec x \in \N^h$ \emph{Lipschitz} if $|x^{k+1}-x^{k}|  \le 1$ for all $k\in [h-1]$.
Given a vector~$\vec x$, denote by $\lip(\vec x)$ the minimal (with respect to domination) Lipschitz vector that dominates~$\vec x$.

Observe that for any Lipschitz vector $\vec x\in \N^h$ and for any Boolean vector $\vec y\in \{0,1\}^h$ we have $\lip(\vec x+\vec y)-\vec x \in \{0,1\}^h$, since $\vec x+\vec{\mathbf{1}}$ is Lipschitz and dominates $\vec x+\vec y$, where $\vec{\mathbf{1}}\in\N^h$ is the all-ones vector.

We now define the firefronts structure corresponding to four Lipschitz vectors $\vec\rho_0, \vec\rho_1, \vec\rho_2, \vec\rho_3$ (see \Cref{fig:front_structure}).

\begin{figure}[h!]
    \centering
    \captionsetup{margin=1cm}
\foreach \Vangle in {65}
{
\tdplotsetmaincoords{\Vangle}{35} \begin{tikzpicture}[line join=round]
        \begin{scope}[tdplot_main_coords]
            \path[get projections];
            \draw [->,very thick, dotted] (\scaling*0.5,\scaling*0.5) -- +(31.6:3cm) -- +(31.9:3.04cm) node[right, scale = 1,yshift=1mm,xshift=-1mm] {$\theta_0$};
            \draw [very thick, dotted] (\scaling*0.5,\scaling*0.5) -- +(211.6:2cm) node[above, scale = 1] {$\theta_2$};
            \draw [->,very thick, dotted] (\scaling*0.5,\scaling*0.5) -- +(-16.2:2.7cm)-- +(-16.27:2.71cm) node[right, scale = 1,yshift=1mm,xshift=-1mm] {$\theta_1$};
            \draw [very thick, dotted] (\scaling*0.5,\scaling*0.5) -- +(163.8:2.5cm) node[below, scale = 1] {$\theta_3$};
            \foreach \z/\a/\b/\c/\d in { 0/ 2/ 5/ -2/ -3,
                                         1/ 1/ 4/ -1/ -3,
                                         2/ 2/ 3/ -1/ -2,
                                         3/ 3/ 2/ -2/ -2
                                         }
            {
\pgfmathsetmacro{\aa}{\a-1}
                
                \foreach \xi in {\d,...,\b}
                {
                        \path (\scaling*\xi,\scaling*\a,\scaling*\z) pic{unit cube};
                }
                \foreach \yi in {\aa,...,\c}
                {
                        \path (\scaling*\d,\scaling*\yi,\scaling*\z) pic{unit cube};
                }
                \foreach \xi in {\d,...,\b}
                {
                        \path (\scaling*\xi,\scaling*\c,\scaling*\z) pic{unit cube};
                }
                \foreach \yi in {\a,...,\c}
                {
                        \path (\scaling*\b,\scaling*\yi,\scaling*\z) pic{unit cube};
                }
            }
\end{scope}
    \end{tikzpicture}
    \hspace{13pt}
    \tdplotsetmaincoords{\Vangle+5}{215} \begin{tikzpicture}[line join=round]
        \begin{scope}[tdplot_main_coords]
            \path[get projections];
            \draw [very thick, dotted] (\scaling*0.5,\scaling*0.5) -- +(27.6:2cm) node[above, scale = 1, yshift=-1mm] {$\theta_2$};
            \draw [->,very thick, dotted] (\scaling*0.5,\scaling*0.5) -- +(207.6:2cm)-- +(207.7:2.01cm) node[above, scale = 1] {$\theta_0$};
            \draw [very thick, dotted] (\scaling*0.5,\scaling*0.5) -- +(-13.7:1.8cm) node[below, scale = 1] {$\theta_3$};
            \draw [->,very thick, dotted] (\scaling*0.5,\scaling*0.5) -- +(167.3:3cm)-- +(167.25:3.01cm) node[below, scale = 1] {$\theta_1$};
            
            \foreach \z/\a/\b/\c/\d in { 0/ 2/ 5/ -2/ -3,
                                         1/ 1/ 4/ -1/ -3,
                                         2/ 2/ 3/ -1/ -2,
                                         3/ 3/ 2/ -2/ -2
                                         }
            {
\pgfmathsetmacro{\cc}{\c+1}
                \foreach \xi in {\b,...,\d}
                {
                        \path (\scaling*\xi,\scaling*\c,\scaling*\z) pic{unit cube};
                }
                \foreach \yi in {\cc,...,\a}
                {
                        \path (\scaling*\b,\scaling*\yi,\scaling*\z) pic{unit cube};
                }
                \foreach \yi in {\c,...,\a}
                {
                        \path (\scaling*\d,\scaling*\yi,\scaling*\z) pic{unit cube};
                }
                \foreach \xi in {\b,...,\d}
                {
                        \path (\scaling*\xi,\scaling*\a,\scaling*\z) pic{unit cube};
                }
            }
        \end{scope}
    \end{tikzpicture}
}
    \caption{
        Two isometric projections of the firefronts structure corresponding to the Lipschitz vectors $\vec\rho_0=(2,1,2,3), \vec\rho_1=(5,4,3,2), \vec\rho_2=(2,1,1,2), \vec\rho_3=(3,3,2,2)$.}
    \label{fig:front_structure}
\end{figure}
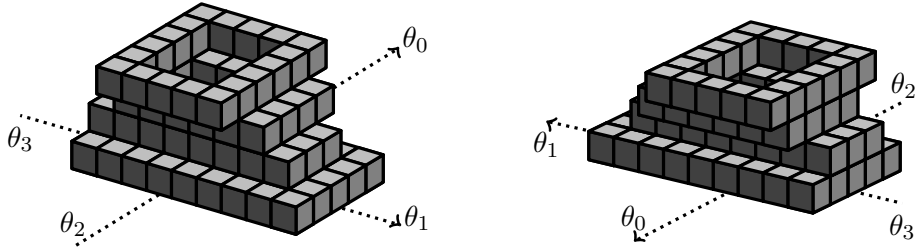 
\begin{defn}\label{def:lDef}
Given four Lipschitz vectors $\rho = (\vec\rho_0, \vec\rho_1, \vec\rho_2, \vec\rho_3)$, where $\vec\rho_i = (r^k_{i})_{k\in[h]}$, 
define the \emph{fronts structure} $(L_0(\rho),L_1(\rho),L_2(\rho),L_3(\rho))$ of $\rho$ to be $L_i(\rho)=\bigsqcup_{k\in[h]} L^k_i(\rho)$, where
\[
  L^k_i(\rho)=L^k_{i,r^k_i}(r^k_{i-1},r^k_{i+1})=\left\{r^k_i\theta_i+m \theta_{i+1}+k\phi\ :\ m\in [-r_{i-1}^k,r_{i+1}^k) \right\}\cap\Z^3.
\]
\end{defn}
Note that $L_i^k(\rho)\cap L_{j}^{\ell}(\rho)= \varnothing$ for $i\ne j\in\I$ and for all $k,\ell\in[h]$.
\subsection{Evolution of the fronts structure}\label{subsec:fronts-evolution}
As explained in \Cref{subsec:pf outline}, we inductively define $\vec\rho_i(t)=\big(r^1_{i}(t),\dots,r^h_{i}(t)\big)$,
 a non-decreasing sequence of Lipschitz vectors whose fronts structure we wish to analyze.
 
Let $\vec\rho(0)$ be some initial quadruple of Lipschitz vectors, and assume without loss of generality that $r_i^k(0)\ge 1$
for all $i\in\I$ and $k\in [h]$.
Given the initial fronts structure $\big(L_i(\vec\rho(0))\big)_{i\in\I}$, define
\begin{equation}\label{eq:r_def}
    \rho_i(t) := \lip\left(\vec\rho_i(t-1)+\vec\alpha_i(t)\right),
\end{equation}
where the Boolean vector $\vec\alpha_i(t)=(a^1_i(t),\ldots,a^h_i(t)) \in \{0,1\}^h$ consists of the \emph{activity indicators} for the levels of the~$i$-th front.
That is, we say that the~$k$-th level of the~$i$-th front is \emph{active} if and only if $a^k_i(t)=1$.
The precise definition of~$\vec\alpha_i(t)$ is delicate, and we postpone this until~\Cref{subs:precise evolution} (see~\Cref{eq:aDef});
for now, it suffices to state that it only depends on the history of the process up until time~$t-1$, and that it satisfies
two conditions related to the notion of fierity to be defined below (see~\Cref{eq:prime a property}).

Note that $\vec\rho_i(t)$ dominates $\vec\rho_i(t-1) + \vec\alpha_i(t)$, and thus $\Delta r_i^k(t) \ge a_i^k(t)$. Moreover, $\vec\rho_i(t-1)$ is Lipschitz and $\alpha_i(t)$ is Boolean so that $\D r_i^k(t) \in \{0,1\}$. In fact, summing over $k\in [h]$ we have
\begin{equation}\label{eq: being pulled}
\D  r_i(t)-a_i(t)=\norm{\vec\rho_i(t)-\left(\vec\rho_i(t-1)+\vec\alpha_i(t)\right)}_1.
\end{equation}
It may be useful for the reader to think of the change in this quantity caused by the $\lip$ operations as ``the $k$-th level of a front being \emph{pulled forward}  by some other active level of that front''.
Consider, for instance, $\vec\rho_1(t-1) = (5,4,3,2)$ of \Cref{fig:front_structure} and suppose that $\vec\alpha_1(t)=(1,0,0,0)$. In this case, only the bottom level is active and yet $\vec\rho_1(t) = (6,5,4,3)$ and $\D r_1(t) - a_1(t) = h-1$.

Henceforth, denote by $L_i^k(t):=L_i^k(\rho(t))$ the~$k$-th level of the~$i$-th front at time~$t$.

\noindent\textbf{Fierity.} Next we define the \emph{fierity}~$\varphi_i^k(t)$ of the~$k$-th level of the~$i$-th front at time~$t$.
\begin{align*}
    \Phi_i^k(t)    &:= L_i^k(t)\cap B(t),\\
    \varphi_i^k(t) &:= \left|\Phi_i^k(t)\right|.
\end{align*}
This quantity measures the number of burning vertices on a particular level. 

The definition of activity indicators in~\Cref{subs:precise evolution} will guarantee the following two fierity-related conditions.
\begin{subequations}\label{eq:prime a property}
\begin{align}
\label{eq:prime a property zero}
a^k_i(t) = 0 \mbox{\quad whenever \quad} \vp^k_i(t-1) &= 0,\\
\label{eq:prime a property one}
a^k_i(t) = 1 \mbox{\quad whenever \quad} \vp^k_i(t-1) &> 4qh^4 - 2.
\end{align}
\end{subequations}
Next, we inductively define~$\D F_i^k(t)$, the hitherto uncounted protected vertices on the~$k$-th level of the~$i$-th front at time~$t$. Setting $F_i^k(0):=\varnothing$, this defines the cumulative counterpart~$F_i^k(t)$. 
\begin{align*}
    \D F_i^k(t) &:= \left((\FF(t)\cap L_i^k(t))\setminus                               B(t)\right)\setminus F^k(t-1);\\
    \D f_i^k(t) &:= |\D F_i^k(t)|.
\end{align*}

Recall that $F^k(t)=\bigsqcup_{i\in\I}F_i^k(t)$ and $F(t)=\bigsqcup_{k\in[k]}F^k(t)$, by our index omitting convention.
Observing that $F(t)\subseteq \FF(t)$, we have $f(t)\le \ff(t)$.
Hence, we may replace the assumption $\ff(t) = ct$ by the bound
$f(t)\leq ct$.
The reader needs not be alarmed by the visual similarity between~$f$ and~$\ff$ (or between~$F$ and~$\FF$) since~$\ff$ and~$\FF$ are no longer required for the lower bound and will thus play no role in~\Cref{sec:Fire growth,,sec:Main proof} (except for a brief appearance in the proof of~\Cref{prop:Vert-Fogarty} in~\Cref{subsec:local-growth}).

\noindent\textbf{Remark.}
The task of selecting~$\vec\alpha_i(t)$ could be thought of as a solitaire sub-game, in which the player takes the challenge of finding a sequence of~$\vec\alpha_i(t)$ satisfying~\Cref{eq:prime a property} such that~$a(t)>0$ infinitely often.

\subsection{Proof of the main theorem}\label{subsec:main-statement}
\Cref{thm:1} naturally consists of an upper bound and a lower bound on $C(G_q)$. In light of the result of~\cite{feldheim20133}, we henceforth assume $h \ge 2$.
The upper bound, which is nothing more than repeating a two-dimensional strategy in each of the~$h$ layers simultaneously, is set by the following claim.
\begin{claim}\label{clm:thm-1-upper}
For $q\in[2]$ and for every $h\ge2$ we have $C(G_q) \le \frac32 qh$.
\end{claim}
\begin{proof}
We need to show that for any~$\varepsilon>0$ and any ${\ff(t)=(\frac32qh +\varepsilon)t}$, there exists a strategy that allows containment of the fire.
Fix~$\varepsilon>0$ and let $\varepsilon\p := \varepsilon/h$. 
Let~$S\subset G_q$ be the initial set of burning vertices. By taking the union of horizontal projections, we can find a finite $S\p\subset \Z^2$ such that $S\subseteq S\p\times[h]$. 
In~\cite[Section 3]{feldheim20133}, the second and third authors show the existence of a strategy capable of containing any finite-source fire in~$\Z\;\Box\;\Z$ (resp.,~$\Z\boxtimes\Z$) for $\ff(t) \ge (\frac32+\varepsilon\p)t$ (resp., $\ff(t) \ge (3+\varepsilon\p)t$). We refer to this strategy, applied to the initial set $S\p$, as the layer strategy.
To achieve the upper bound for~$G_q$, apply the layer strategy in each layer in parallel to obtain a containment strategy on~$G_q$ for
\[
        \ff (t) 
    \ge \left(\frac{3qh}2 + h\varepsilon\p\right)t 
    =   \left(\frac{3qh}2 + \varepsilon\right)t.\qedhere
\]
\end{proof}

The lower bound is expressed by the somewhat technical~\Cref{lem:technical}. To state it, we require the following constant.
Let
\[
    \lambda := h + q\max\{r_0(0) + r_2(0), r_1(0) + r_3(0)\}.
\]
In light of~\Cref{eq:Lab-length}, $\lambda$ is an upper bound on the total length of any two initial firefronts.
\begin{lem}\label{lem:technical}
    Let $q\in[2]$, and suppose that
    \begin{equation}\label{eq:large-vp0}
        \vp(0) \ge \lambda + 55h^5
    \end{equation}
    and that
    \begin{equation}\label{eq:small-f}
        f(t) \le \frac{3qh}2t \text{ for all $t\in\N$.}
    \end{equation}
    Then $\vp(t)+f(t)\ge 3qht$ for all $t \in \N$.
\end{lem}
\begin{proof}[Reduction of the lower bound in~\Cref{thm:1} to \Cref{lem:technical}]
Set $r_i^k(0) = 30h^4$ for all $k\in[h]$ and $i\in\I$, and pick the initial burning set to be $B(0) = L(0) = \sqcup_{i\in\I} L_i(0)$.
Observe that $r_i(0)=30h^5$ and $\vp_i(0) = |L_i(0)| = 30q h^5$ for all $i\in\I$, so indeed
\[
    120qh^5 = \vp(0) \ge \lambda + 55h^5 = h + 60qh^5 + 55h^5.
\]
Set $c=\frac32qh$ and apply~\Cref{lem:technical} to deduce that $\vp(t)+f(t) \ge 2ct$ for all $t\in\N$. 
Now $\vp(t)\ge ct$ since $f(t)\le ct$, and by the definition of~$\vp(t)$ we obtain $|B(t)| \ge \vp(t) \ge ct$.
Thus the fire expands indefinitely and is never contained.
\end{proof}

To prove~\cref{lem:technical} we use an inductive argument. 
The induction hypothesis is
\[
    \text{\hyp{t}}: \quad\vp(t) + f(t) \ge 3qht + \lambda + 52h^5.
\]
The base case~\hyp{t} for $t\le 2h$ is implied by~\Cref{eq:large-vp0} since $3h^5 > 3qht$ for all $t\le 2h$. It thus remains to prove the following proposition, which is the inductive step. This will be done in~\Cref{subsec:induction-step}.
\begin{prop}\label{prop:inductive-step}
    Fix $t > 2h$.
    Assuming~\Cref{eq:large-vp0} and~\Cref{eq:small-f}, if~\hyp{s} holds for all $s<t$ then~\hyp{t} holds.
\end{prop}
\section{Fire growth}\label{sec:Fire growth}

In this section we bound from below the increase in the number of
burning vertices on~$L_i^k(t)$. 

\subsection{Shifted vertices and the potential of the fronts structure}
Towards defining the potential~$\mu_i^k(t)$, we need to keep track of vertices that shifted from the $i$-th firefront to an adjacent front $j=i\pm1$.
To see how such a change might occur, consult~\Cref{fig:fog2}. 

For $i\in\I$, $k\in[h]$ and $t\in\N$ let
\begin{align*}
    V^k_{i,i-1}(t) &:= (L_{i-1}^k(t)\setminus L_{i-1}^k(t-1))\cap L_{i}^k(t-1)\\
    V^k_{i,i+1}(t) &:= (L_{i+1}^k(t)\setminus L_{i+1}^k(t-1))\cap\left(L_{i}^k(t-1)\right)^+
\end{align*}
The reader should keep in mind that each of $V^k_{i,i-1}(t)$ and $V^k_{i,i+1}(t)$ always consists of at most one vertex; formally, this will be established as part of~\Cref{prop:old&newFog}.

We may now define $p_i^k(t)$, which counts the total change in the number of vertices on the $k$-th level of the $i$-th front as the result of transition of vertices between this front and adjacent fronts.

First, define $p_{i,j}^k(t)$ for $i,j\in\I$ satisfying $|i-j|=1$ as follows. Set $p_{i,j}^k(0):=0$ and let $p_{i,j}^k(t) := p_{i,j}^k(t-1) + \D p_{i,j}^k(t)$, where
\[
    \D p_{i,j}^k(t):=\left|V_{i,j}^k(t) \cap \left(B(t)\cup F(t)\right)\right|-\left|V_{j,i}^k(t) \cap \left(B(t)\cup F(t)\right)\right|.
\]
Next, let 
\begin{eqnarray*}
    p^k_i(t) &:=& p^k_{i,i+1}(t)+ p^k_{i,i-1}(t);\\
    \D p^k_i(t) &:=& \D p^k_{i,i+1}(t) + \D p^k_{i,i-1}(t).
\end{eqnarray*}
The skew-symmetry in the definition of $\D p_{i,j}^k(t)$ leads to cancellations when summing over all directions $i\in\I$.
Indeed,
\begin{align}\label{eq: sum of dp is 0}
        \D p^k(t)
    &=  \sum_{i\in\I}\left(\D p_{i,i+1}^k(t) 
        + \D p_{i,i-1}^k(t)\right)
    =   \sum_{i\in\I}\D p_{i,i+1}^k(t)
        + \sum_{i\in\I}\D p_{i,i-1}^k(t)\notag\\
	&=  \sum_{i\in\I}\D p_{i,i+1}^k(t)
	    + \sum_{i\in\I}\D p_{i+1,i}^k(t)
	=   0,  
\end{align}
and thus also $\D p(t) = 0$ and $p(t) = 0$, for all $t\in\N$.

We now have all the ingredients to define the \emph{potential} of the $k$-th level of the $i$-th front as
\begin{equation}\label{muDef}
\mu_i^k(t) := \vp_i^k(t)+f_i^k(t)+p^k_i(t).
\end{equation}
Note that $\mu(t) = \vp(t) + f(t)$ by the skew-symmetry observation above, and that for $t=0$ we simply get $\mu_i(0) = \vp_i(0)$.

We now present two potential-related bounds, whose proofs are postponed to~\Cref{subsec:potential bounds}.
The first states that $\D \mu_i^k(t)$ is always non-negative, and grows by~$q$ whenever the level is active.
\begin{prop}\label{prop:mu>a>0}
    For all $t\in \N$, $k\in[h]$ and $i\in\I$
    we have $\D\mu_i^k(t)\ge qa_i^k(t)$.
\end{prop}
In particular, this implies
\begin{equation}\label{eq:mu-increasing}
    \mu_i^k(t) \ge \mu_i^k(s) \text { for all $t\ge s$.}
\end{equation}
The second bound states that the potential of each firefront exceeds its fierity minus a constant.
\begin{prop}\label{prop:mu>vp}
    For all $t\in\N$ and $i\in\I$
    we have $\mu_i(t) > \vp_i(t) - 4qh^5$.
\end{prop}
 
\subsection{Local fire growth bounds}\label{subsec:local-growth}
\noindent\textbf{Growth within layers.}
The following bound is a slight generalization of a result of Fogarty~\cite[Theorem 1]{fogarty2003catching}. 
To better comprehend this variant, please see~\Cref{fig:fog1}.
\begin{figure}[h!]
    \captionsetup{margin=1cm}
    \subfloat[$\scriptstyle{q=2},{a=2},{b=2}$,\\\phantom{(a) }$\scriptstyle{\delta_-=\delta_+=1}$\label{fig:a1} ]{
    \begin{tikzpicture}[scale=0.8]
        \draw[step=1cm,grayX,ultra thin] (0,-3) grid (3,4);
        \draw [->] (0.5,0.5) -- +(0,0.3);
        \draw [->] (0.5,0.5) -- +(0.3,0);
        \draw [draw=none,darkgrayX, fill=darkgrayX] (1,-1) rectangle (2,3);
        \draw [black, very thick, dashed,pattern=my north east lines,pattern color=darkgrayX] (2,-2) rectangle (3,4);
        \draw [black, very thick, dotted] (0,2.5) -- (3,2.5);
        \draw [darkgrayX, very thick, dotted] (0,-1.5) -- (3,-1.5);
        \draw [grayX, very thick, dotted] (0,3.5) -- (3,3.5);
        \draw [grayX, very thick, dotted] (0,-2.5) -- (3,-2.5);

        \pgfmathsetmacro{\op}{0.5};

        \foreach \y in {0,2}{
            \node at (1.5,-0.5+\y) {\Fire};
        }
        \node at (1.5,-0.5+1) {\faShield};
        \foreach \y in {-1,0}{
            \node at (2.5,-0.5+\y) {\faShield};
        }
        \foreach \y in {1,...,3}{
            \node at (2.5,-0.5+\y) {\Fire};
        }
        \end{tikzpicture}}
    \hfill
    \subfloat[$\scriptstyle{q=2},{a=2},{b=2}$,\\\phantom{(a) }$\scriptstyle{\delta_-=\delta_+=0}$\label{fig:a2}]{
    \begin{tikzpicture}[scale=0.8]
        \draw[step=1cm,grayX,ultra thin] (0,-3) grid (3,4);
        \draw [->] (0.5,0.5) -- +(0,0.3);
        \draw [->] (0.5,0.5) -- +(0.3,0);
        \draw [draw=none,darkgrayX, thick,fill=darkgrayX] (1,-1) rectangle (2,3);
        \draw [black, thick, dashed,pattern=my north east lines,pattern color=grayX] (2,-1) rectangle +(1,4);
        \draw [black, thick, dashed] (2,-2) rectangle +(1,1);
        \draw [black, thick, dashed] (1,2) rectangle +(1,1);
        \node [draw,scale=0.5] at (1.75,2.25) {$-$};
        \node [draw,scale=0.5] at (2.75,-1.75) {$+$};
        \draw [darkgrayX, very thick, dotted] (0,2.5) -- (3,2.5);
        \draw [darkgrayX, very thick, dotted] (0,-1.5) -- (3,-1.5);
    \end{tikzpicture}}
    \hfill
    \subfloat[$\scriptstyle{q=1},{a=3},{b=3}$,\\\phantom{(a) }$\scriptstyle{\delta_-=\delta_+=0}$\label{fig:a3}]{
        \includegraphics{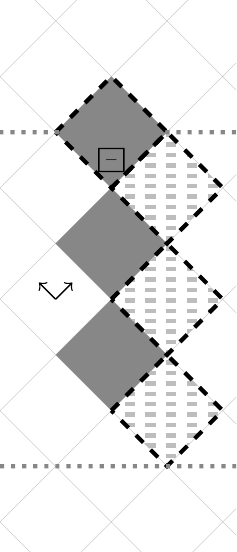}}
    \hfill
    \subfloat[$\scriptstyle{q=1},{a=3},{b=3}$,\\\phantom{(a) }$\scriptstyle{\delta_-=\delta_+=1}$\label{fig:a4}]{
        \includegraphics{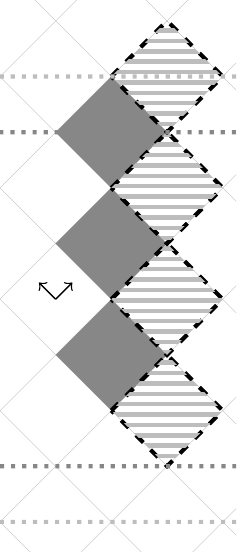}}
    \hfill
    \subfloat[$\scriptstyle{q=1},{a=4},{b=2}$,\\\phantom{(a) }$\scriptstyle{\delta_-=\delta_+=0}$\label{fig:a5}]{
        \includegraphics{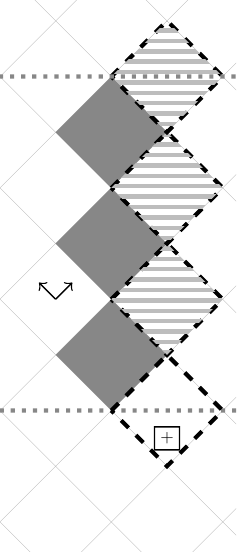}}
    \hfill
    \subfloat[$\scriptstyle{q=1},{a=4},{b=2}$,\\\phantom{(a) }$\scriptstyle{\delta_-=\delta_+=1}$\label{fig:a6}]{
        \includegraphics{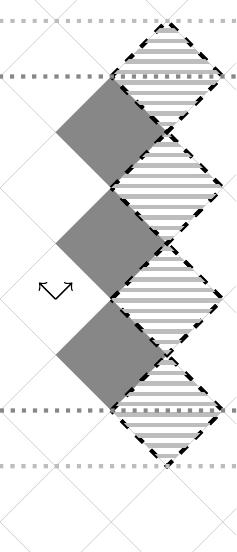}}
    \caption{The transition from $L^k_{1,1}(a,b)$ to $L^k_{1,2}(a+\delta_-,b+\delta_+)$ for different values of $q$, $a$, $b$, $\delta_-$ and $\delta_+$.
        Actual spread of burning vertices is illustrated in sub-figure~(\ref{fig:a1}).
        $L^k_{1,1}(a,b)$ vertices are colored with a darker shade, while $L^k_{1,2}(a+\delta_-,b+\delta_+)$ bear a lighter striped color. 
        The origin and the Cartesian axes are marked by short arrows, and when $V_-$ (resp.,~$V_+$) is non-empty, its vertex is tagged by $\boxminus$ (resp.,~$\boxplus$).
        Dark dotted lines mark $-a\theta_0+\R\theta_1+k\phi$ and $b\theta_0+\R\theta_1+k\phi$, while light dotted lines mark $-(a+\delta_-)\theta_0+\R\theta_1+k\phi$ and $(b+\delta_+)\theta_0+\R\theta_1+k\phi$.
Observe, for $q=1$, how the parity of $d+a$ and $d+b$ affects the shape of $L^k_{1,2}(a+\delta_-,b+\delta_+)$, and the cardinality of $V_-$ and $V_+$.}
    \label{fig:fog1}
\end{figure}
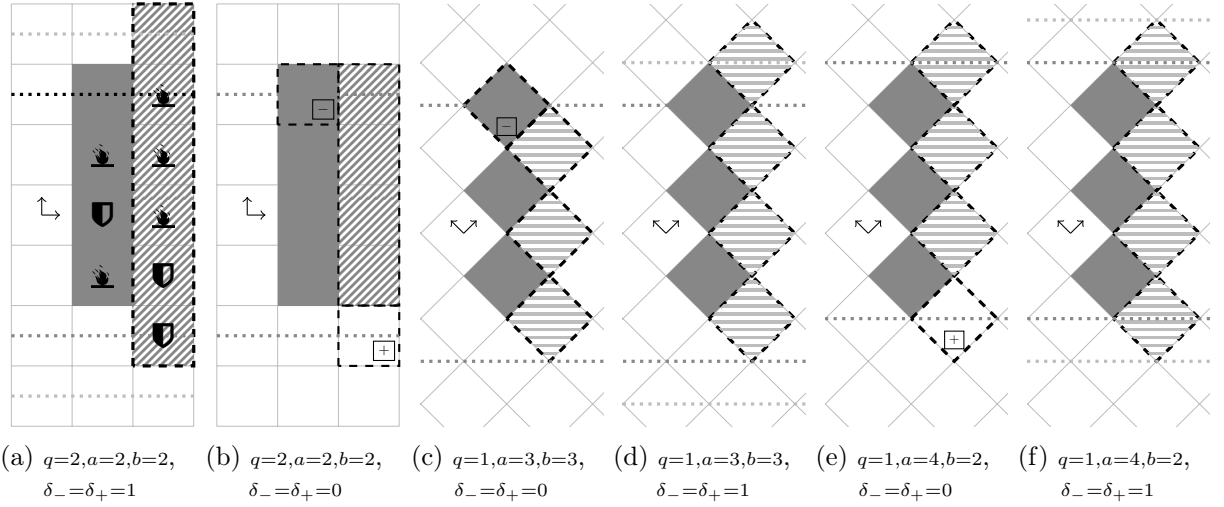 
\begin{prop}\label{prop:fog1}
Let $q\in\{1,2\}$, $i\in\I$, $k\in[h]$, and suppose that $\left|L^k_{i,d}(a,b)\cap B(t-1)\right|>0$ in~$G_q$ for some $a,b,d\in\N$.
Then, for $\delta_{-},\delta_{+}\in \{0,1\}$ we have
\[
        \left|\left(L^k_{i,d+1}(a+\delta_{-},b+\delta_{+})
                \sqcup V_{-}\sqcup V_{+}\right)
            \cap \left(B(t)\cup \FF(t)\right)\right|
    \ge \left|L^k_{i,d}(a,b)\cap B(t-1)\right| + q,
\]
where
\begin{align*}
    V_{-} &= 
    \begin{cases}
        \left(\left\{d\theta_i-a\theta_{i+1}+k\phi\right\}
            \cap \Z^3\right) & \delta_{-}=0,\\
        \varnothing & \delta_{-}=1,
    \end{cases}\\
    V_{+} &=
    \begin{cases}
        \left(\left\{(d+1)\theta_i+b\theta_{i+1}+k\phi\right\}
            \cap \Z^3\right) & \delta_{+}=0,\\
        \varnothing & \delta_{+}=1.
    \end{cases}
\end{align*}
\end{prop}

\begin{proof} Write $A=L^k_{i,d}(a,b)\cap B(t-1)$ and observe that $A\neq\varnothing$ by our assumption.
Denote by~$x_{\max}$ (resp.,~$x_{\min}$) the element $d\theta_i+m\theta_{i+1}+k\phi\in A$ with the maximal (resp., minimal) value of~$m$. Let
\[
    N = (\Z^2\times \{k\}) \cap A^+ \cap L_{i,d+1}^k(a+1,b+1)
\]
and note that $\theta_i \pm \theta_{i+1}\in \Z^3$ for both $q=1$ and $q=2$, while $\theta_i\in \Z^3$ only for $q=2$.
It follows that $|N| \ge |A|+q$, since
\[
    N \supseteq \{x+\theta_i+\theta_{i+1}\ :\ x \in A\}
        \sqcup \{x_{\min}+\theta_i-\theta_{i+1}\}
        \sqcup \left(\{x_{\min}+\theta_i\}\cap\Z^3\right).
\]
This resolves the case $\delta_-=\delta_+=1$.

If $x_{\min}\neq d\theta_i-a\theta_{i+1}+k\phi$, then $N \subseteq L_{i,d+1}^k(a,b+1)$ (see~\Cref{fig:a1,,fig:a4,,fig:a5,,fig:a6});
otherwise, $V_-\subseteq A \subseteq B(t)$ (see~\Cref{fig:a3,fig:a2}). This resolves the case $\delta_->\delta_+$.
Similarly, if $x_{\max}\neq d\theta_i+(b-1)\theta_{i+1}+k\phi$,
then $N\subseteq L_{i,d+1}^k(a+1,b)$ (see~\Cref{fig:a1,,fig:a3,,fig:a4,,fig:a6});
otherwise, $V_+$ is non-empty and its sole member is adjacent to $x_{\max}$ so that
$V_+\subseteq B(t)\cup F(t)$ (see \Cref{fig:a5,fig:a2}). This resolves the case $\delta_-<\delta_+$.
Since $L_{i,d+1}^k(a,b+1)\cap L_{i,d+1}^k(a+1,b)=L_{i,d+1}^k(a,b)$, which resolves the last case $\delta_-=\delta_+=0$.
\end{proof}

To apply~\Cref{prop:fog1} for the fronts structure we require some technical computations, provided in the following proposition.
Please see~\Cref{fig:fog2} for a more intuitive perspective.
\begin{figure}[h!]
    \captionsetup{margin=1cm}
    \subfloat[$\scriptstyle{q=2}$,
    \\\phantom{(a) }$\scriptstyle{r_{i- 1}(t-1)=2}$,
    \\\phantom{(a) }$\scriptstyle{r_{i+ 1}(t-1)=2}$,
    \\\phantom{(a) }$\scriptstyle{\D r_{i\pm 1}(t)=1}$                                                                                     
    \label{fig:b1} ]{
\begin{tikzpicture}[scale=0.7]
        \pgfmathsetmacro{\ymin}{-2};
        \pgfmathsetmacro{\yclip}{\ymin-1.5};
        \clip (-0.3,\yclip) rectangle (3.25,4.5);
        \draw[step=1cm,grayX,ultra thin] (-3,-5) grid (4,5);
        \draw [->] (0.5,0.5) -- +(0,0.3);
        \draw [->] (0.5,0.5) -- +(0.3,0);
        
\pgfmathsetmacro{\temp}{\ymin+1}
\draw [draw=none,fill=darkgrayX] (1,\temp) rectangle (2,3);
        \draw [black, very thick, dashed] (2,\ymin) rectangle (3,4);
        \draw [draw=none,pattern=my north east lines,pattern color=darkgrayX] (2,\ymin) rectangle (3,4);

\draw [draw=none,fill=grayX] (-2,2) rectangle (1,3);
        \draw [black, very thick, dashed] (-2,3) rectangle (2,4);
        \draw [draw=none,pattern=my north west lines,pattern color=grayX] (-2,3) rectangle (2,4);

\pgfmathsetmacro{\temp}{\ymin-1}
        \draw [draw=none,fill=grayX] (-2,\ymin) rectangle +(4,1);
        \draw [black, very thick, dashed] (-2,\temp) rectangle +(5,1);
        \draw [draw=none,pattern=my north west lines,pattern color=grayX] (-2,\temp) rectangle +(5,1);
    \end{tikzpicture}}
    \hfill
    \subfloat[$\scriptstyle{q=2}$,
    \\\phantom{(a) }$\scriptstyle{r_{i- 1}(t-1)=2}$,
    \\\phantom{(a) }$\scriptstyle{r_{i+ 1}(t-1)=2}$,
    \\\phantom{(a) }$\scriptstyle{\D r_{i\pm 1}(t)=0}$
        \label{fig:b2}]{
    \begin{tikzpicture}[scale=0.7]
        \clip (-0.3,-3.5) rectangle (3.25,4.5);
        \draw[step=1cm,grayX,ultra thin] (-3,-5) grid (4,5);
        \draw [->] (0.5,0.5) -- +(0,0.3);
        \draw [->] (0.5,0.5) -- +(0.3,0);

\draw [draw=none,fill=darkgrayX] (1,-1) rectangle (2,3);
        \draw [black, thick, dashed] (2,-1) rectangle (3,3);
        \draw [draw=none,pattern=my north east lines,pattern color=darkgrayX] (2,-1) rectangle (3,3);

\draw [draw=none,fill=grayX] (-2,2) rectangle (1,3);
        \draw [black, very thick, dashed] (-2,2) rectangle (2,3);
        \draw [draw=none,pattern=my north west lines,pattern color=grayX] (-2,2) rectangle (2,3);

        \draw [draw=none,fill=grayX] (-2,-2) rectangle +(4,1);
        \draw [black, very thick, dashed] (-2,-2) rectangle +(5,1);
        \draw [draw=none,pattern=my north west lines,pattern color=grayX] (-2,-2) rectangle +(5,1);

        \node [draw,scale=0.5] at (1.75,2.25) {$-$};

        \pgfmathsetmacro{\sqx}{2+0.75};
        \pgfmathsetmacro{\sqy}{-2+0.25};
        \node [draw,scale=0.5] at (\sqx,\sqy) {$+$};
    \end{tikzpicture}}
    \hfill
    \subfloat[$\scriptstyle{q=1}$,
    \\\phantom{(a) }$\scriptstyle{r_{i-1}(t-1)=3}$,
    \\\phantom{(a) }$\scriptstyle{r_{i+1}(t-1)=3}$,
    \\\phantom{(a) }$\scriptstyle{\D r_{i\pm 1}(t)=1}$
        \label{fig:b3}]{
            \includegraphics{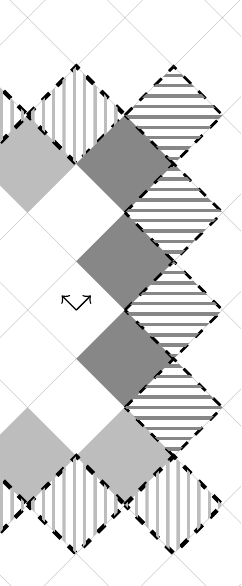}}
    \hfill
    \subfloat[$\scriptstyle{q=1}$,
    \\\phantom{(a) }$\scriptstyle{r_{i-1}(t-1)=3}$,
    \\\phantom{(a) }$\scriptstyle{r_{i+1}(t-1)=3}$,
    \\\phantom{(a) }$\scriptstyle{\D r_{i\pm 1}(t)=0}$
        \label{fig:b4}]{
            \includegraphics{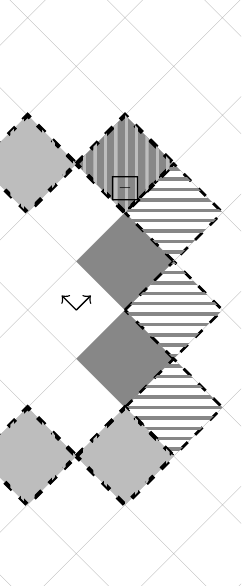}}
    \hfill
    \subfloat[$\scriptstyle{q=1}$,
    \\\phantom{(a) }$\scriptstyle{r_{i-1}(t-1)=4}$,
    \\\phantom{(a) }$\scriptstyle{r_{i+1}(t-1)=2}$,
    \\\phantom{(a) }$\scriptstyle{\D r_{i\pm 1}(t)=1}$
        \label{fig:b5}]{
            \includegraphics{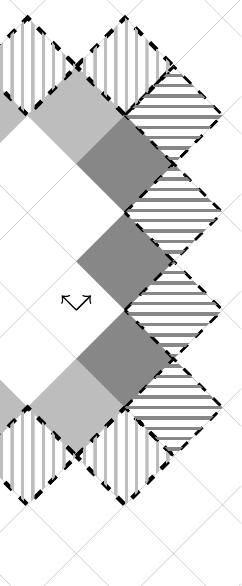}}
    \hfill
    \subfloat[$\scriptstyle{q=1}$,
    \\\phantom{(a) }$\scriptstyle{r_{i-1}(t-1)=4}$,
    \\\phantom{(a) }$\scriptstyle{r_{i+1}(t-1)=2}$,
    \\\phantom{(a) }$\scriptstyle{\D r_{i\pm 1}(t)=0}$
        \label{fig:b6}]{
            \includegraphics{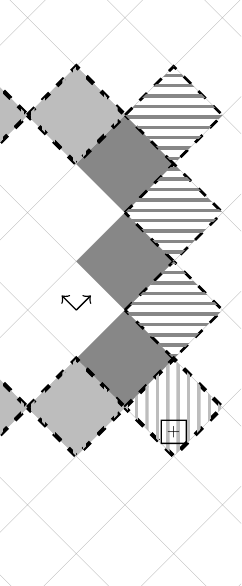}}
    \caption{
    The interaction between $L_{i-1}^k, L_{i+1}^k$ and $L_i^k$ for time $t-1$ and $t$, as a result of 
    different parities of $r_i^k(t-1)+r_{i-1}^k(t-1)$, $r_i^k(t-1)+r_{i+1}^k(t-1)$ and different values of $\D r_{i-1}^k(t)$, $\D r_{i+1}^k(t)$. 
    $L_i^k(t-1)$ and $L_i^k(t)$ are indicated by dark squares and stripes, respectively, while $L_{i\pm 1}^k(t-1)$ and $L_{i\pm 1}^k(t)$ are indicated by light squares and stripes. The origin and the Cartesian axes are marked by short arrows.
    Observe that $L_{i\pm 1}^k(t-1)$
    and $L_i^k(t)$ follow a different striped pattern.
    Since $\Delta r_i^k(t-1)=1$ in all sub-figures, $\D p_{i,i\pm 1}^k(t)$ are always non-negative and $\D p_{i\pm 1,i}^k(t)$ are non-positive.
    When $\D p_{i,i\pm 1}^k(t)$ take non-zero value, the vertex of $V_{i,i\pm 1}^k(t-1)$ is tagged by~$\boxminus$ or~$\boxplus$.}
    \label{fig:fog2}
\end{figure}
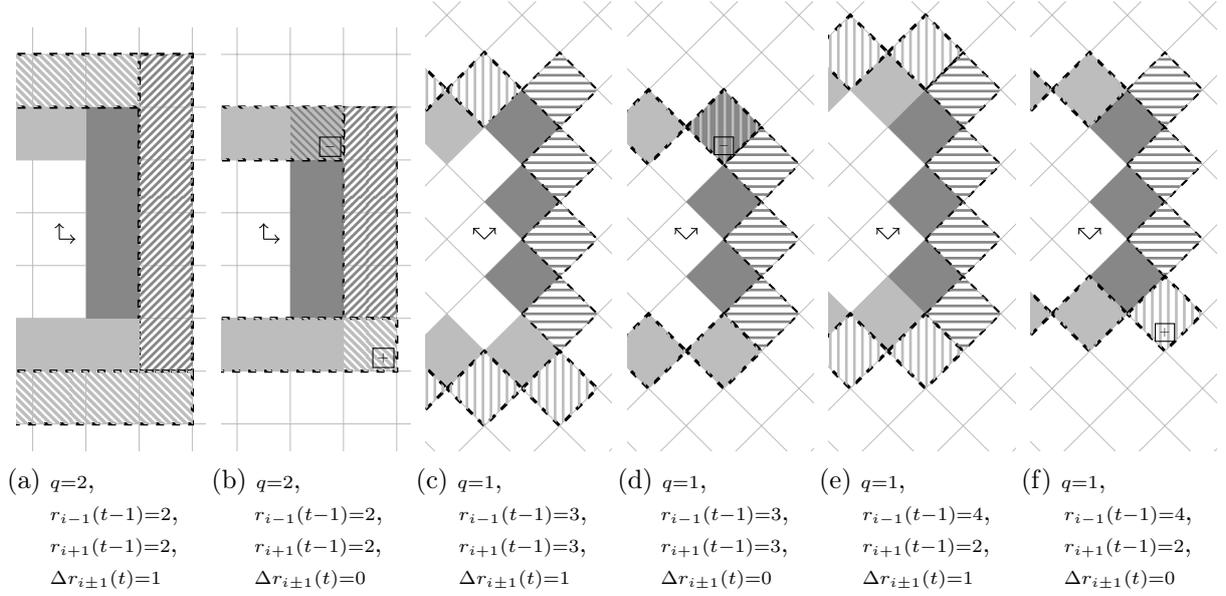 
\begin{prop}\label{prop:old&newFog}
    For $d=r_i^k(t-1)$, $d_{\pm}=r_{i\pm1}^k(t-1)$, $\delta=\D r_{i}^k(t)$, $\delta_{\pm}=\D r_{i\pm 1}^k(t)$ we have
    \begin{equation}\label{eq:oldnewfogeq1}
        L_{i,d}^k(d_-,d_+)=L_i^k(t-1), \qquad\qquad\qquad L_{i,d+\delta}^k(d_-+\delta_-,d_++\delta_+)=L_i^k(t).
    \end{equation}
    \begin{align}
    V^k_{i,i-1}(t)&=\begin{cases}
        \left(\left\{d\theta_i^{(q)}-d_-\theta_{i+1}^{(q)}+k\phi\right\}\cap \Z^3\right) & \delta(1-\delta_{-})=0,\\
        \mathrlap{\varnothing}\hphantom{\left(\left\{(d+1)\theta_i^{(q)}+d_+\theta_{i+1}^{(q)}+k\phi\right\}\cap \Z^3\right)} & \delta(1-\delta_{-})=1,
    \end{cases}\label{eq:oldnewfogeq2a}\\
    V^k_{i,i+1}(t)&=\begin{cases}
        \left(\left\{(d+1)\theta_i^{(q)}+d_+\theta_{i+1}^{(q)}+k\phi\right\}\cap \Z^3\right)  & \delta(1-\delta_{+})=1,\\
        \varnothing & \delta(1-\delta_{+})=0.
    \end{cases}\label{eq:oldnewfogeq2b}
    \end{align}
\end{prop}
\begin{proof}
First, observe that~\Cref{eq:oldnewfogeq1} is immediate from~\Cref{def:lDef}, recalling that $L_i^k(t)=L_i^k(\rho(t))$. 
Next, to see~\Cref{eq:oldnewfogeq2a} and~\Cref{eq:oldnewfogeq2b},
observe that if $\D r_{i-1}^k(t)=\delta_{-}=1$ then $d_-<r_{i-1}^k(t)$,
so that $L_{i-1}^k(t)\cap L_i^k(t-1)=\varnothing$ (see~\Cref{fig:b1,fig:b3,fig:b5}). 
If $\D r_{i-1}^k(t)=0$, then, by definition, an element is in both $L_{i-1}^k(t)$ and $L_i^k(t-1)$ if and only if it is of the form 
\begin{align*}
        r_i^k(t-1)\theta_i-      r_{i-1}^k(t)                         \theta_{i+1}+k\phi
    &=  r_i^k(t-1)\theta_i-\left(r_{i-1}^k(t-1)+\D r_{i-1}^k(t)\right)\theta_{i+1}+k\phi \\
    &=  d\theta_i-d_-\theta_{i+1}+k\phi,
\end{align*}
and $\delta_-=0$ (see~\Cref{fig:b2,fig:b4}).
Because $L_{i-1}^k(t)$ and $L_i^k(t)$ are disjoint, we must have $d\theta_i-d_-\theta_{i+1}+k\phi \notin L_{i-1}^k(t-1)$.

Similarly, if $\D r_{i+1}^k(t)=\delta_{+}=1$ then $d_+<r_{i+1}^k(t)$, so that $L_{i+1}^k(t)\cap L_i^k(t-1)=\varnothing$,
whereas if $\D r_{i+1}^k(t)=0$ then an element is in both $L_{i+1}^k(t)$ and $\big(L_i^k(t-1)\big)^+$ if it is of the form 
\[
        \left(r_{i}^k(t-1)+\beta\right)\theta_i+r_{i+1}^k(t-1)\theta_{i+1}+k\phi
    =   \left(d+\beta\right)\theta_i+d_+\theta_{i+1}+k\phi
\]
with $\beta\in \{0,\pm1\}$ (see~\Cref{fig:b2,fig:b6}).
However, we have
\begin{multline*}
    \left\{(d+\beta)\theta_i+d_+\theta_{i+1}+k\phi\, :\, \beta\in \{0,\pm1\}\right\} \setminus L_{i+1}^k(t-1)\\
    =\begin{cases}
        \left\{(d+1)\theta_i+d_+\theta_{i+1}+k\phi\right\}&\D r_i^k(t)=1,\\
        \varnothing&\D r_i^k(t)=0,
    \end{cases}
\end{multline*}
and the proposition follows.
\end{proof}

\Cref{prop:old&newFog} implies the following, which serves a role in the proof of~\Cref{prop:mu>a>0}.
\begin{cor}\label{obs:invasion}
    Let $t\in\N$, $k\in[h]$ and $i,j\in\I$ such that $|i-j|=1$.
    Then $\D r_i^k(t)=0$ implies $V_{i,j}^k(t) = \varnothing$, while $\D r_i^k(t)=1$ implies $V_{j,i}^k(t) = \varnothing$.
\end{cor}
In particular,  for all $t\in\N$, $k\in[h]$ and $i\in\I$ we have
\begin{equation}\label{eq:p-small}
    |\D p^k_i(t)| \le 2
\end{equation}
since $\D p^k_{i,i\pm1}(t)\in\{-1,0,1\}$.
Furthermore, by~\Cref{obs:invasion} we obtain the following, which is used to prove~\cref{prop:mu>vp}.
\begin{cor}\label{cor:p>0}
	Let $t\in\N$, $k\in[h]$ and $i\in\I$.
	Then $\D r_i^k(t)=1$ implies $\D p_i^k(t) \ge 0$.
\end{cor}

\noindent\textbf{Vertical growth.} The inter-layer graph connectivity of~$G_q$ implies the following.
For $\ell=k\pm 1$, we have
\begin{equation}\label{obs:overflowing}
        \left|L^\ell_{i,d}(a,b)\cap \left(B(t)\cup \FF(t)\right)\right|
    \ge \left|L^k_{i,d}(a,b)\cap B(t-1)\right|.
\end{equation}

\begin{prop}\label{prop:Vert-Fogarty}
    Let $i\in \I$ and $k,\ell\in [h]$ such that $|k-\ell| = 1$,
    and assume that $\D r_i^\ell(t) = 1$ and $r_i^k(t-1) = r_i^\ell(t-1) + 1$. 
    Then $\vp_i^\ell(t) + \D f_i^\ell(t) \ge \vp_i^k(t-1) - 2$.
\end{prop}
\begin{proof}
We have $L_i^k(t-1) = L^k_{i,d}(d_-,d_+)$,
where $d = r_i^k(t-1) = r_i^\ell(t)$ and $d_\pm = r_{i\pm1}^k(t-1)$.
Moreover, by the Lipschitz property of~$\vec\rho_{i\pm 1}(t-1)$ we have
\[
        d_\pm 
    =   r_{i\pm 1}^k(t-1)
    \le 1 + r_{i\pm 1}^{\ell}(t-1)
    \le 1 + r_{i\pm 1}^{\ell}(t),
\]
so that
\[
    L_i^\ell(t) \supseteq L_{i,d}^\ell(d_- - 1, d_+ - 1)
\]
and thus
\begin{align*}
         \vp_i^{\ell}(t) + \D f_i^{\ell}(t)
    &=   \left|\Phi_i^{\ell}(t) \sqcup \D F_i^{\ell}(t)\right|
\\
    &=   \left|L_i^{\ell}(t) \cap 
            \left(B(t) \cup \FF(t)\right)\right| \\
    &\ge \left|L_{i,d}^\ell(d_- - 1,d_+ - 1) \cap 
            \left(B(t) \cup \FF(t)\right)\right| \\
    &\ge \left|L_{i,d}^\ell(d_-,d_+)\cap 
            \left(B(t) \cup \FF(t)\right)\right| - 2 \\
    &\ge \vp_i^k(t-1) - 2,
\end{align*}
where the last inequality is by~\cref{obs:overflowing}.
\end{proof}
\subsection{Potential bounds' proofs}\label{subsec:potential bounds}
\begin{proof}[Proof of~\Cref{prop:mu>a>0}]
We consider two cases separately.
First, if $\D r_i^k(t)=0$ then $a_i^k(t) = 0$, so we need to show that $\D\mu_i^k(t)\ge 0$. Now
\[
        \vp_i^k(t)+\D f_i^k(t)
    =   \left|\Phi_i^k(t) \sqcup \D F_i^k(t)\right|
    =   \left|\Phi_i^k(t-1)\right| 
        + \left|\D\Phi_i^k(t) \sqcup \D F_i^k(t)\right|,
\]
since $\Phi_i^k(t-1)\subseteq \Phi_i^k(t)$.
Moreover, $V_{i,i-1}^k(t)=V_{i,i+1}^k(t)=\varnothing$ by \cref{obs:invasion}, so 
\[
\D p_i^k(t)
    =   \D p_{i,i-1}^k(t) + \D p_{i,i+1}^k(t)
    =    - \left|\left(V_{i-1,i}^k(t)
                \sqcup V_{i+1,i}^k(t)\right)
            \cap \left(B(t)\sqcup F(t)\right)\right|.
\]
Putting these together we obtain
\begin{align*}
         \D\mu_i^k(t) 
    &=   \D\vp_i^k(t) + \D f_i^k(t) + \D p_i^k(t) \\
    &=   \left|\D\Phi_i^k(t) \sqcup \D F_i^k(t)\right| 
         - \left|\left(V_{i-1,i}^k(t)
                \sqcup V_{i+1,i}^k(t)\right)
            \cap \left(B(t)\sqcup F(t)\right)\right|
    \ge  0,
\end{align*}
since $V_{i,i\pm1}^k(t)\subseteq L_{i\pm1}^k(t) \setminus L_{i\pm1}^k(t-1)$.

Next, if $\D r_i^k(t)=1$ we have $V_{i-1,i}^k(t)=V_{i+1,i}^k(t)=\varnothing$ by~\cref{obs:invasion}, so that
\begin{align*}
        \vp_i^k(t)+\D p^k_i(t)
    &=  |L_i^k(t)\cap B(t)|+\D p^k_{i,i+1}(t)+\D p^k_{i,i-1}(t)\\
    &=  \Big|\Big(L_i^k(t)\sqcup\big(V_{i,i-1}^k(t)\sqcup     V_{i,i+1}^k(t)\big)\Big)\cap B(t)\Big| + \Big|\big(V_{i,i-1}^k(t)\sqcup V_{i,i+1}^k(t)\big)\cap F(t)\Big|.
\end{align*}
Note that as $\D r_i^k(t)=1$, we have $\D F_i^k(t)=(L_i^k(t)\setminus B(t))\cap F(t)$ so that showing
\[
        \D\mu_i^k(t)
    =   \D\vp_i^k(t)+\D p_i^k+\D f_i^k
    \ge qa_i^k(t)
\] 
reduces to proving that
\[
        \left|\Big(L_i^k(t) \sqcup \big(V_{i,i-1}^k(t) \sqcup V_{i,i+1}^k(t)\big)\Big) \cap \big(B(t)\sqcup F(t)\big)\right|
    =   \vp_i^k(t)+\D p_i^k+\D f_i^k
    \ge qa_i^k(t) + \vp_i^k(t-1).
\]
When $\vp_i^k(t-1)=0$ we have $a_i^k(t)=0$ by~\Cref{eq:prime a property zero}, and this is straightforward;
otherwise, this follows from~\Cref{prop:fog1}, using~\Cref{prop:old&newFog}.
\end{proof}

\begin{proof}[Proof of~\Cref{prop:mu>vp}]
It suffices to show that for all $k\in[h]$
\[
     p_i^k(t)+f_i^k(t) = \mu_i^k(t) - \vp_i^k(t) > -4qh^4.
\]
Fix $t\in\N$ and $k\in[h]$. 
If $\D r_i^k(\tau)=1$ for all $\tau\in[t]$, then
\[
    p_i^k(t) = \sum_{\tau\in[t]}\D p_i^k(\tau) \ge 0
\]
by~\Cref{cor:p>0}, so $p_i^k(t)+f_i^k(t) \ge 0$ and we are done;
otherwise, let
\[
s = \max\{\tau\in[t]\ :\ \D r_i^k(\tau)=0\}.
\]
By the definition of $s$ we have $\D r_i^k(s)=0$, so that $a_i^k(s)=0$ and, by~\Cref{eq:prime a property one}, $\vp_i^k(s-1) \le 4qh^4 - 2$.

Next we verify the following inequalities.
\begin{align*}
    f_i^k(t) &\ge f_i^k(s-1), \\
    p_i^k(t) &\ge p_i^k(s), \\
    \mu_i^k(s-1) & > 0.
\end{align*}
The first follows from the fact that $\D f_i^k(\tau)\ge 0$ for all~$\tau\in\N$;
the second --- 
from the fact that
$\D p_i^k(\tau)\ge 0$ for $s<\tau\le t$
which, in turn, follows from~\Cref{cor:p>0} and the choice of~$s$;
the last follows from~\cref{eq:mu-increasing} and the fact that $\mu_i^k(0) = \vp_i^k(0)>0$.

Putting these three inquealities together, we obtain
\begin{align*}
         f_i^k(t) + p_i^k(t)
    &\ge f_i^k(s-1) + p_i^k(s)\\
    &=   \mu_i^k(s-1) - \vp_i^k(s-1) + \D p_i^k(s)\\
    &\ge 0 - (4qh^4 - 2) - 2
    =   -4qh^4,    
\end{align*}
where the last inequality is by~\Cref{eq:p-small}.
\end{proof} 
\section{Proof of the main technical statement}\label{sec:Main proof}

\subsection{Precise evolution of the fronts structure}\label{subs:precise evolution}
Finally we are ready to present the definition of the activity indicators~$a_i^k(t)$, which concludes the exact description of the evolution of the fronts structure.
This is done using an auxiliary indicator function~$g_i(t)$.
The role of~$g_i(t)$ is to indicate whether the~$i$-th front have experienced a period of reduced growth of its potential while having significant fierity (in comparison with the natural~$qh$ growth per time-step, guaranteed by~\Cref{prop:mu>a>0} when $a_i^k(t)=1$ for all~$k\in[h]$).
This definition guarantees that leaving such a period (that is, to have~$g_i(t)=0$ and~$g_i(t-1)=1$) requires either an increase of~$2qh^2$ in the potential, or a reduction of the fierity of the front below a threshold.

Define
\begin{equation}\label{eq:gDef}    
    g_i(t) := \begin{cases} 1 & 
    \Big(\vp_i(t) > 4qh^5\Big)\text{ and } \Big(\D\mu_i(t) < qh\left(1-g_i(t-1)\right\} + 2qh^2g_i(t-1)\Big)\\
    0 & \text{otherwise.}
    \end{cases}
\end{equation}
Using this, we define $a_i^k(t)$ as
\begin{equation}\label{eq:aDef}    
    a_i(t) := \begin{cases} 1 & 
    \varphi^k_i(t-1) >
    \left(4h\left(1-g_i(t-1)\right) + 4qh^3g_i(t-1)\right)\left(r_i^k(t-1) - r_i^{\min}(t-1)\right)\\
    0 & \text{otherwise.}
    \end{cases}
\end{equation}
It is easy to verify that both~\Cref{eq:prime a property zero} and~\Cref{eq:prime a property one} hold.

Equipped with the precise definition of $a_i^k(t)$, we state three auxiliary lemmata, whose proofs we postpone to~\Cref{subsec:aux-lemmata-proofs}.
\Cref{lem:PullAndOverflow} establishes that whenever a level is ``pulled'', its potential growth is typical. \Cref{lem:badTimes} establishes that the duration of a period of reduced potential growth is at most~$2h$ time steps. Finally, \Cref{lem:g_i(t)<8h^5} establishes that a long period of high fierity can be subdivided into small pieces, such that during each piece (except the last), the average potential growth is typical.
\begin{lem}\label{lem:PullAndOverflow}
    Let $t\in\N$, $i\in\I$, and suppose that $\D r_i(t)>a_i(t)$.
    Then $g_i(t)=0$ and $\D \mu_i(t) \ge qh$.
\end{lem}
\begin{lem}\label{lem:badTimes}
    Let $t\in\N$, $i\in\I$, and suppose that $\D \mu_i(t+1) < qh$ and $g_i(t)=0$.
    Then there exists $t+1 \le t\p \le t+2h$ with $g_i(t\p)=0$.
\end{lem}
\begin{lem}\label{lem:g_i(t)<8h^5}
    Let $s,t\in\N$ such that $s<t$ and let $i\in\I$.
    Suppose that $g_i(s)=0$ and that for all $s\le\tau\le t+2h$ we have $\vp_i(\tau) > 4qh^5$.
    Then there exists $t\le t\p\le t+2h$ such that $g_i(t\p)=0$ and 
    \[\mu_i(t\p)-\mu_i(s) \ge qh(t\p-s).
    \]\end{lem}
The following proposition extends~\Cref{prop:mu>a>0}; it relates the potential growth with the expansion of the fronts structure.
\begin{prop}\label{prop:mu>r}
    For all $t\in\N$ and $i\in\I$ we have $\D \mu_i(t) \ge q\D r_i(t)$.
\end{prop}
\begin{proof}
Fix $i\in\I$ and $t\in\N$.
If $\D r_i(t)=a_i(t)$ then we are done by~\Cref{prop:mu>a>0}.
Otherwise, we have $\D r_i(t) > a_i(t)$ and thus 
$\D \mu_i(t) \ge qh$ by~\Cref{lem:PullAndOverflow}.
But $qh\ge q\D r_i(t)$ since $\D r_i^k(t)\in\{0,1\}$ for all $k\in[h]$.
\end{proof}
Note that $\D r_i(t)\ge a_i(t)$ and thus~\Cref{prop:mu>r} is indeed stronger than what is obtained by summing~\Cref{prop:mu>a>0} over all $k\in[h]$.
It implies the following.
\begin{cor}
    For all $t\in\N$ and $i\in\I$ we have
    \begin{equation}\label{eq:vp-upper-bound}
        2\vp_i(t) \le \lambda - \vp_{i-1}(0) - \vp_{i+1}(0) + \mu_{i-1}(t) + \mu_{i+1}(t)
    \end{equation}
\end{cor}
\begin{proof}
Apply~\Cref{prop:mu>r} for $j = i\pm1$, and sum over $\tau\in[t]$. We get
\[
    q(r_j(t)-r_j(0)) \le \mu_j(t)-\mu_j(0) = \mu_j(t)-\vp_j(0).
\]
Next, observe that $\Phi_i^k(t) \subseteq L_i^k(t)$ and use~\Cref{eq:Lab-length} to obtain
\[
    2\vp_i^k(t) \le 1 + q\left(r_{i-1}^k(t) + r_{i+1}^k(t)\right).
\]
Summing over $k\in[h]$ and plugging in the definition of~$\lambda$
establishes the result.
\end{proof}

\subsection{The fierity of two firefronts}
The following proposition, which is reminiscent of~\cite[Lemma~9]{feldheim20133},
establishes that, under suitable assumptions, the sum of the fierity of any two firefronts at two nearby times is somewhat significant.
\begin{prop}\label{prop:fierity-of-two}
    Let $i,j\in\I$ such that $i\neq j$, let $s\le t\le s+2h$  and suppose that \Cref{eq:large-vp0} and~\Cref{eq:small-f} hold, and~\hyp{s} is satisfied.
    Then $\vp_i(s)+\vp_j(t) > 8qh^5$.
\end{prop}    
\begin{proof}
We establish the proposition using the following inequality, which we verify below.
\begin{equation}\label{eq:vp-f upper}
    \vp(s) - f(s) < \lambda + 8qh^5 + 2\vp_i(s) + 2 \vp_j(s).
\end{equation}
Indeed, \hyp{s} implies
\[
    \vp(s) + f(s) \ge \lambda + 52h^5 + 3qhs,
\]
and since
\[
    52h^5 \ge 24qh^5 + 32h^2 > 24qh^5 + 8h^2 + 3qh(t-s),
\]
we have
\begin{align*}
         \vp(s) - f(s)
    &\ge \lambda + 24qh^5 + 8h^2 + 3qht  - 2f(s) \\
    &\ge \lambda + 24qh^5 + 8h^2 + 2f(t) - 2f(s),
\end{align*}
where the last inequality is by~\cref{eq:small-f}.
Using~\Cref{eq:vp-f upper}, we get
\[
        \vp_i(s) + \vp_j(s)
    >   8qh^5 + 4h^2 + f(t) - f(s) 
    \ge 8qh^5 + 4h^2 + f_j(t) - f_j(s).
\]
Thus it remains to show that
\[
    \vp_j(t) + f_j(t) + 4h^2 \ge \vp_j(s) + f_j(s).
\]
This is implied by~\cref{eq:mu-increasing} and~\cref{eq:p-small} since
\[
        \D\vp_j(\tau) + \D f_j(\tau)
    =   \D\mu_j(\tau) - \D p_j(\tau)
    \ge -2h
\]
for all $s<\tau\le t$ and thus
\[
        \vp_j(t) + f_j(t)
    \ge \vp_j(s) + f_j(s) - 2h(t-s)
    \ge \vp_j(s) + f_j(s) - 4h^2.
\]
Now we verify~\Cref{eq:vp-f upper} for the two possible cases.

\noindent\textbf{Adjacent fronts.}
Without loss of generality, assume $i=0$ and $j=1$.
Use \Cref{eq:vp-upper-bound} on the third and fourth fronts
to obtain
\[
        2\vp_2(s) + 2\vp_3(s)
    \le 2\lambda - \vp(0) + \mu(s)
    \le \lambda + \mu(s),
\]
where $\vp(0) \ge \lambda$ by~\Cref{eq:large-vp0}.
Subtracting both sides from $2\vp(s) + \lambda$ we get
\[
        \lambda + 2\vp_0(s) + 2\vp_1(s)
    =   \lambda + 2\vp(s) - \left(2\vp_2(s) + 2\vp_3(s)\right)
    \ge 2\vp(s) - \mu(s)
    =   \vp(s) - f(s),
\]
which implies~\Cref{eq:vp-f upper}.

\noindent\textbf{Opposing fronts.}
Without loss of generality assume $i=0$ and $j=2$.
Use~\Cref{eq:vp-upper-bound} for the second and fourth fronts
to obtain
\[
        \vp_1(s) + \vp_3(s)
    \le \lambda - \vp_0(0) - \vp_2(0) + \mu_0(s) + \mu_2(s)
    \le \lambda + \mu_0(s) + \mu_2(s),
\]
since $\vp_0(0),\vp_2(0)\ge 0$.
Next use~\Cref{prop:mu>vp} for the second and fourth fronts to obtain
\[
    \vp_1(s) + \vp_3(s) < \mu_1(s) + \mu_3(s) + 8qh^5.
\]
Altogether,
\[
    2\vp_1(s) + 2\vp_3(s) < \lambda + \mu(s) + 8qh^5,
\]
which implies~\Cref{eq:vp-f upper} also in this case.
\end{proof}
\subsection{Proof of the induction step}\label{subsec:induction-step}
\begin{proof}[Proof of~\cref{prop:inductive-step}]
The proof relies on an induction on the following property.
We say that $(s_0, s_1, s_2, s_3)$ \emph{control} time~$s$ if these satisfy $s \le s_i \le s+4h$ and $g_i(s_i)=0$ for all $i\in\I$, $|s_i-s_j|\le 2h$ for all $i,j\in\I$, and 
\begin{equation}\label{eq:mu local bound}
        \sum_{i\in\I}\mu_i(s_i) 
    \ge \vp(0) - qhs_{\min} + \sum_{i\in\I}qhs_i.
\end{equation}
First note that if time $t-2h$ is controlled by some $(t_0, t_1, t_2, t_3)$ then~\hyp{t} holds.
Indeed, by~\cref{eq:mu-increasing} and~\Cref{eq:mu local bound} we have
\begin{align*}
         \mu(t)
    &=   \sum_{i\in\I}\mu_i(t)
    \ge  \sum_{i\in\I}\mu_i(t_i)
    \ge  \vp(0) - qht_{\min} + \sum_{i\in\I}qht_i
    \ge  \vp(0) + 3qht_{\min} \\
    &\ge \vp(0) + 3qh(t-2h)
    >    \vp(0) + 3qht - 2h^5 >    \lambda + 3qht + 52h^5,
\end{align*}
where the last inequality is due to~\Cref{eq:large-vp0}.
Thus~\hyp{t} holds.

We prove by induction that each time $y\le t-2h$ is controlled by some $(y_0,y_1,y_2,y_3)$.
Clearly time $y=0$ is controlled by~$(0,0,0,0)$; it remains to show that if~\hyp{s} holds for all $s<t$ and time~$y$ is controlled, then time~$y+1$ is also controlled. To do so, we show that if $(y_0, y_1, y_2, y_3)$ control time~$y$, either they control time~$y+1$ too, or there exist $(s_0, s_1, s_2, s_3)$ with $s_i\ge y_i$ and at least one~$i$ such that $s_i>y_i$ that control either time~$y$ or time~$y+1$.
By iterating this finitely many times, the result will follow.
    
To this end, assume that $(y_0, y_1, y_2, y_3)$ control time $y$.
If $y < y_{\min}$, then $(y_0, y_1, y_2, y_3)$ also control $y+1$. Assume thus, without loss of generality, that $y_0 = y = y_{\min}$. 

If $\D \mu_0(y_0+1)\ge qh$, then it is easy to verify that $g_0(y_0+1) = 0$ and $\{y_0+1,y_1,y_2,y_3\}$ satisfy~\cref{eq:mu local bound}, so $\{y_0+1,y_1,y_2,y_3\}$ control time $y$.
Otherwise, $\D\mu_0(y_0+1)< qh$ so, by \Cref{lem:badTimes}, there is $y_0 < s_0 < y_0+2h$ with $g_0(s_0)=0$.
We consider two cases, corresponding to the two reasons in 
\cref{eq:gDef} to have $g_0(s_0)=0$.
First, in the case
\begin{equation}\label{eq: first case of g is 0}
    \mu_0(s_0)-\mu_0(y_0) \ge 2qh^2\ge qh\left(s_0-y_0\right),
\end{equation}
use \Cref{eq:mu local bound} for $(y_0,y_1,y_2,y_3)$ to obtain that
\begin{align*}
         \mu_0(s_0)+\sum_{j\in [3]}\mu_j(y_j)
    &=  \mu_0(s_0)-\mu_0(y_0)+\sum_{i\in\I}\mu_i(y_i) \\
    &\ge qh(s_0-y_0) + \vp(0) -qhy_0 +\sum_{i\in\I}qhy_i \\
    &\ge \vp(0) -qh\min\{s_0,y_1,y_2,y_3\} + qhs_0+\sum_{j\in [3]}qhy_j
\end{align*}
and thus $(s_0,y_1,y_2,y_3)$ control time~$y$.

If \Cref{eq: first case of g is 0} does not hold, then $\vp_0(s_0) < 4qh^5$ by \Cref{eq:gDef}.
For $j\in[3]$, set $s_j = y_j$ if $y_j\ge s_0$;
otherwise, apply \Cref{lem:g_i(t)<8h^5} with $(y_j,s_0)$ 
to deduce the existence of $s_0\le s_j<s_0+2h$ with $g_j(s_j) = 0$ that satisfies
\[
        \mu_j(s_j)-\mu_j(y_j)
    \ge qh(s_j-y_j).
\]
The conditions for~\Cref{lem:g_i(t)<8h^5} are satisfied since $g_j(y_j)=0$ and for all $s_0-2h \le y_j \le \tau \le s_0 + 2h$ we may apply~\Cref{prop:fierity-of-two} for $s=\min\{s_0,\tau\}$ and $\max\{s_0,\tau\}$ to obtain 
\[
    \vp_j(\tau) > \vp_j(\tau) + \vp_0(s_0) - 4qh^5 \ge 8qh^5 - 4qh^5 = 4qh^5.
\]
Note that $s\le s_0 < y_0+2h < t$ so~\hyp{s} indeed holds.

Next, use \Cref{eq:mu local bound} for $(y_0,y_1,y_2,y_3)$ to obtain that
\begin{align*}
         \mu_0(y_0) + \sum_{j\in[3]}\mu_j(s_j)
    &=   \sum_{j\in\I}\mu_j(y_j) + \sum_{j\in[3]}(\mu_j(s_j)-\mu_j(y_j))\\
    &\ge \vp(0) + \sum_{j\in[3]} qhy_j + \sum_{j\in[3]}qh(s_j-y_j)
    =    \vp(0) + \sum_{j\in[3]} qhs_j.
\end{align*}
Now $\mu_0(y_0) \le \mu_0(s_0)$ since $y_0<s_0$, so
\[
        \sum_{i\in\I}\mu_i(s_i)
    \ge \vp(0) + \sum_{j\in[3]} qhs_j
    =   \vp(0) - qhs_{\min} + \sum_{i\in\I} qhs_i,
\]
as $s_0=s_{\min}$.
This establishes that time $y+1$ is controlled by $(s_0,s_1,s_2,s_3)$, which concludes the proof.
\end{proof} 
\subsection{Proofs of auxiliary lemmata}\label{subsec:aux-lemmata-proofs}
\begin{proof}[Proof of~\Cref{lem:PullAndOverflow}]
Recall that $\Phi_i^k(t)=L_i^k(t)\cap B_i^k(t)$. First we show that there exist $k,\ell\in[h]$, $|k-\ell| = 1$ such that 
\begin{align*}
    r_i^k(t)&=r_i^k(t-1)+1 &&r_i^k(t-1)=r_i^{\ell}(t-1)+1&& r_i^{\ell}(t)=r_i^{\ell}(t-1)+1\\
a_i^k(t)&=1 &&&&a_i^{\ell}(t)=0
\end{align*} 
To see this, recall the definition of $r_i^k(t)$ given in \Cref{eq:r_def} and let $\vec v_i  := \vec\rho_i(t) - \left(\vec\rho_i(t-1)+\vec\alpha_i(t)\right)$.
Observe that $\vec v_i$ is not the all-zero vector by the assumption $\norm{\vec v_i} = \D r_i(t)-a_i(t) > 0$.
We select $\ell\in[h]$ among the non-zero coordinates of $\vec v_i$ and an adjacent $k = \ell \pm 1$ for which $a_i^k(t) = 1$.
By the Lipschitz property, such $k$ and $\ell$ exist.

Observe that $a_i^k(t)=1,a_i^{\ell}(t)=0$, so by~\Cref{eq:aDef},
\begin{eqnarray*}
    \vp_i^k(t-1) &>& \left(4h\left(1-g_i(t-1)\right)+4qh^3g_i(t-1)\right)\left(r_i^k(t-1)-r_i^{\min}(t-1)\right);\\
    \vp_i^{\ell}(t-1) &\le& \left(4h\left(1-g_i(t-1)\right)+4qh^3g_i(t-1)\right)\left(r_i^{\ell}(t-1)-r_i^{\min}(t-1)\right).
\end{eqnarray*}
Since $r^k_i(t-1)=r^{\ell}_i(t-1)+1$ we get
\[
        \vp_i^k(t-1) - \vp_i^\ell(t-1)
    >   4h\left(1-g_i(t-1)\right)+4qh^3g_i(t-1) 
    \ge 4h.
\]
Together with $\vp_i^{\ell}(t)+\D f_i^{\ell}(t) \ge \vp_i^k(t-1)-2$, obtained by~\cref{prop:Vert-Fogarty}, we have
\[
      \D\mu_i^{\ell}(t)
    = \D\vp_i^{\ell}(t) + \D f_i^{\ell}(t) + \D p_i^\ell(t)
    > 4h - 4,
\]
where $\D p_i^\ell(t)\ge -2$ by~\cref{eq:p-small}.
Now
\[
        \D\mu_i(t) 
    \ge \D\mu_i^\ell(t)
    >    4h - 4
    \ge qh,
\]
using $h\ge 2$, and the proposition follows.
\end{proof}

\begin{proof}[Proof of~\Cref{lem:badTimes}]
Apply \Cref{prop:mu>a>0} together with $\D \mu_i(t+1)<qh$ to obtain $a_i(t+1)<h$, which ensures the existence of $\ell \in [h]$ for which $a_i^\ell(t+1)=0$. This implies  
\begin{equation}\label{eq:vp i ell bnd}
        \vp_i^\ell(t)
    \le 4h\left(r_i^\ell(t)-r_i^{\min}(t)\right)
    \le 4h(h-1),
\end{equation}
by \Cref{eq:aDef} and our assumption $g_i(t)=0$.
Let $t\p=\min\{\tau>t\ :\ g_i(\tau)=0\}$.
We need to show $t\p\le t+2h$.
Note that 
\begin{equation}\label{eq:vp i ell tau bnd}
    \D\vp_i^\ell(\tau)
    \le \D\mu_i^\ell(\tau) - \D p_i^\ell(\tau)
    \le \D\mu_i(\tau) + 2
    \le 2qh^2 + 2
\end{equation}
for all $t+1\le\tau\le t'$ by~\cref{eq:p-small} and~\cref{eq:gDef}.
We consider three cases separately. 

\textbf{First}, we consider the case in which there exists $t+2 \le \tau \le t+2h-1$ for which $a_i^\ell(\tau)=1$, denoting by $T$ the minimal such time.
If $t\p < T-1$ we are done;
otherwise, by~\Cref{eq:vp i ell bnd} and~\Cref{eq:vp i ell tau bnd} we have
\begin{align*}
       \vp_i^\ell(T-2)
  =    \vp_i^\ell(t)+\sum_{\tau=t+1}^{T-2}\D \vp_i^\ell(\tau)
  &\le 4h(h-1) + (2qh^2+2)(T-t-2) \\
  &\le 4h(h-1) + (2qh^2+2)(2h-3)
  <    4qh^3 - 2qh^2 - 2,
\end{align*}
while $\vp_i^\ell(T-1) > 4qh^3$ (by~\Cref{eq:aDef}).
Hence $\D\vp_i^\ell(T-1) > 2qh^2+2$.
Using again the fact that $\D p_i^\ell(T-1)\ge -2$, we obtain, by~\Cref{eq:gDef}, that $t\p=T-1$. 

\textbf{Second}, we consider the case in which there is a $k\in[h]$ such that $a_i^k(\tau)=1$ for every $t+1 \le \tau \le t+2h-1$. This implies
\[
    r_i^k(t+2h-1) - r_i^k(t) = 2h-1.
\]
By the Lipschitz property of $\vec\rho_i(t)$ and $\vec\rho_i(t+2h-1)$ we have
\[
    r_i^{\ell}(t)-(h-1)\le r_i^k(t)\text{\quad and \quad}
    r_i^k(t+2h-1)-(h-1)\le r_i^\ell(t+2h-1),
\]
so
\[
        r_i^{\ell}(t+2h-1)-r_i^{\ell}(t)
    \ge r_i^k(t+2h-1)-r_i^k(t)-2(h-1) \ge 1.
\]
Let $T=\min\{\tau\ge t\ :\ \D r_i^\ell(\tau)=1\}$ and note that $a_i^\ell(T) = 0 < \D r_i^\ell(T)$. Now $\D r_i(T)>a_i(T)$ and by~\Cref{lem:PullAndOverflow} we obtain $t\p=T$.

\textbf{Lastly}, in the remaining case, for each $k\in[h]$ there exists $t+1 \le \tau^k \le t+2h-1$ satisfying $a_i^k(\tau^k)=0$. 
In this case, by~\Cref{eq:aDef} we have $\vp_i^k(\tau^k)< 4qh^3(h-1)$. 
If $t\p<\tau^{\max}$ we are done;
otherwise, by~\cref{eq:vp i ell tau bnd} we have $\D\vp_i^k(\tau)\le 2qh^2 + 2$ for $\tau^k<\tau\le\tau^{\max}$. 
Thus
\begin{align*}
         \vp_i(\tau^{\max}) 
    &=   \sum_{k\in[h]}\vp_i^k(\tau^k)
           +\left(\vp_i^k(\tau^{\max})-\vp_i^k(\tau^k)\right) \\
    &\le \sum_{k\in[h]}4qh^3(h-1)+(2qh^2+2)(\tau^{\max}-\tau^k) \\
    &\le \sum_{k\in[h]}4qh^3(h-1)+(2qh^2+2)(2h-2) \\
    &=   4h(h-1)(h+1)(qh^2+1) < 4qh^5,
\end{align*}
which gives $t\p\le \tau^{\max}+1$.
The proposition follows.
\end{proof}

\begin{proof}[Proof of~\Cref{lem:g_i(t)<8h^5}]
For all $s\le\tau\le t+2h$ we have
\begin{equation}\label{eq:simpler-g}
    g_i(t) := \begin{cases} 1 & 
    \D\mu_i(\tau) < qh(1-g_i(\tau-1)) + 2qh^2 g_i(\tau-1)\\
    0 & \text{otherwise.}
    \end{cases}
\end{equation}
by our assumption $\vp_i(\tau) > 4q^5$ and the definition of~$g_i$ in~\cref{eq:gDef}.
Let $\tau_0=s$ and inductively set $\tau_n=\min\{\tau>\tau_{n-1}\ :\ g_i(\tau)=0\}$ for $n\ge 1$.
We show that $t\p=\tau_m$ satisfies the required property, where $m=\min\{n\in\N\ :\ \tau_n > t\}$.

First note that $\tau_n\le \tau_{n-1} + 2h$ for all~$n\in[m]$, and in particular $t\p\le \tau_{m-1} + 2h \le t + 2h$.
Indeed, if $\tau_n > \tau_{n-1}+1$, then $g_i(\tau_{n-1}+1) = 1$, and, by~\cref{eq:simpler-g}, $\D\mu_i(\tau_{n-1}+1) < qh$. Thus we may apply~\Cref{lem:badTimes} with~$\tau_{n-1}$ to obtain $\tau_n \le \tau_{n-1} + 2h$.
Next, $g_i(\tau_n)=0$ for all $n\in[m]$ so
\[
    \D\mu_i(\tau_n)\ge qh(1-g_i(\tau_n-1))+2qh^2g_i(\tau_n-1)
\]
by~\cref{eq:simpler-g}.
Now either $\tau_n-1=\tau_{n-1}$ so that
\[
    \D \mu_i(\tau_n) \ge qh = qh(\tau_n-\tau_{n-1}),
\]
or $g_i(\tau_n-1)=1$, in which case
\[
    \D \mu_i(\tau_n) \ge 2qh^2 \ge qh(\tau_n-\tau_{n-1}).
\]
In either case, by~\cref{eq:mu-increasing} this implies 
\[
        \mu_i(\tau_n)-\mu_i(\tau_{n-1})
    \ge qh\left(\tau_n-\tau_{n-1}\right).
\]
Summing this for $n\in[m]$ we obtain
\[
        \mu_i(t\p)-\mu_i(s)
    =   \mu_i(\tau_m)-\mu_i(\tau_0)
    \ge qh(\tau_m-\tau_0)
    =   qh(t\p-s),
\]
and the lemma follows.
\end{proof}

\bibliography{Main}
\bibliographystyle{abbrv}
\end{document}